\documentclass[onefignum,onetabnum]{siamonline181217}
\usepackage{lipsum}
\usepackage{amsfonts}
\usepackage{graphicx}
\usepackage{subcaption}
\usepackage{graphbox}
\usepackage{epstopdf}
\usepackage{algorithmic}
\usepackage{booktabs}
\usepackage{dsfont}
\usepackage[margin=1.25in]{geometry}
\usepackage[colorinlistoftodos, disable]{todonotes} % remove the 'disable' to get the  todonotes to show up again: , disable
\usepackage{xcolor}
\usepackage{bbm}
\usepackage{mathtools}
\usepackage{amsmath}
\usepackage{bm}
\usepackage{verbatim}
\usepackage{wrapfig}
\usepackage[noadjust]{cite}
\usepackage{tikz}
\usepackage{soul}
\renewenvironment{equation*}{\[}{\]\ignorespacesafterend}

\usetikzlibrary{shapes,arrows}
\usetikzlibrary{decorations.pathreplacing,angles,quotes}
\tikzstyle{block1} = [rectangle, draw, thick,fill=blue!10, text width=4.5em, text centered, rounded corners, minimum height=2em, minimum width = 5cm]

\tikzstyle{block2} = [rectangle, draw, thick,fill=blue!10, text width=2em, text centered, rounded corners, minimum height=2em]

\tikzstyle{line} = [draw, -latex']
\newcommand{\arvind}[1]{\todo[color=blue!10,size=\tiny]{#1}}
\newcommand{\note}[1]{\todo[color=red!10,size=\tiny]{#1}}

\newcommand{\dena}[1]{\todo[color=green!10,size=\tiny]{#1}}

\ifpdf
  \DeclareGraphicsExtensions{.eps,.pdf,.png,.jpg}
\else
  \DeclareGraphicsExtensions{.eps}
\fi

\usepackage{enumitem}
\setlist[enumerate]{leftmargin=.5in}
\setlist[itemize]{leftmargin=.5in}

\newsiamremark{remark}{Remark}
\newsiamremark{hypothesis}{Hypothesis}
\crefname{hypothesis}{Hypothesis}{Hypotheses}
\newsiamthm{claim}{Claim}
\newcommand{\mcA}{{\mathcal{A}}}

\newcommand{\mcK}{{\mathcal{K}}}

\newcommand{\mcP}{{\mathcal{P}}}
\newcommand{\mcM}{{\mathcal{M}}}

\newcommand{\Ff}{\mathcal{F}}

\newcommand{\Gg}{\mathcal{G}}

\newcommand{\PP}{\mathbb{P}}
\newcommand{\EE}{\mathbb{E}}

\newcommand{\RR}{\mathbb{R}}
\newcommand{\N}{\mathbb{N}}
\newcommand{\ikrep}{{(k)}}
\newcommand{\ovc}{\overrightarrow{C}}

\newcommand{\mfT}{{\mathfrak{T}}}
\newcommand{\mfN}{{\mathfrak{N}}}

\newcommand{\R}{{\mathds{R}}}
\renewcommand{\arvind}[1]{{}}

\newtheorem{assumption}{Assumption}

\graphicspath{{figures/}}

\DeclareFontFamily{U}{mathx}{\hyphenchar\font45}
\DeclareFontShape{U}{mathx}{m}{n}{
      <5> <6> <7> <8> <9> <10>
      <10.95> <12> <14.4> <17.28> <20.74> <24.88>
      mathx10
      }{}
\DeclareSymbolFont{mathx}{U}{mathx}{m}{n}
\DeclareMathSymbol{\bigtimes}{1}{mathx}{"91}

\headers{MFG-PA Games in REC markets}
{}

\title{Principal agent mean field games in Renewable Energy Certificate (REC) markets}
\author{Dena Firoozi\thanks{Department of Decision Sciences, HEC Montréal, Montreal, QC, Canada (\email{dena.firoozi@hec.ca}})
\and Arvind V Shrivats\thanks{Department of Operations Research \& Financial Engineering, Princeton University, Princeton, NJ, USA (\email{shrivats@princeton.edu}})
\and Sebastian Jaimungal\thanks{Department of Statistical Sciences, University of Toronto, Toronto, ON, Canada (\email{sebastian.jaimungal@utoronto.ca}}}

\usepackage{amsopn}

\ifpdf
\hypersetup{
  pdftitle={Principal agent mean field games in REC markets},
  pdfauthor={Dena Firoozi, Arvind V Shrivats, Sebastian Jaimungal}\\
 }

\fi
\begin{document}

\maketitle
% REQUIRED
\begin{abstract}

Principal agent games are a growing area of research which focuses on the optimal behaviour of a principal and an agent, with the former contracting work from the latter, in return for providing a monetary award. While this field canonically considers a single agent, the situation where multiple agents, or even an infinite amount of agents are contracted by a principal are growing in prominence and pose interesting and realistic problems. Here, agents form a Nash equilibrium among themselves, and a Stackelberg equilibrium between themselves as a collective and the principal. 
We apply this framework to the problem of implementing Renewable Energy Certificate (REC) markets, where the principal requires regulated firms (power generators) to pay a non-compliance penalty which is inversely proportional to the amount of RECs they have. RECs can be obtained by generating electricity from clean sources or purchasing on the market. The agents react to this penalty and optimize their behaviours to navigate the system at minimum cost. In the agents' model we incorporate market clearing as well as agent heterogeneity. For a given market design, we find the Nash equilibrium among agents using techniques from mean field games. We then use techniques from extended McKean-Vlasov control problems to solve the principal (regulators) problem, who aim to choose the penalty function in such a way that balances environmental and revenue impacts optimally. We find through these techniques that the optimal penalty function is linear in the agents' state, suggesting the optimal emissions regulation market is more akin to a tax or rebate, regardless of the principal's utility function.
\end{abstract}

\begin{keywords}
principal agent games, mean field games, emissions markets, market design, Nash equilibrium, Stackelberg equilibrium, FBSDE, stochastic control
\end{keywords}

% REQUIRED
\begin{AMS}
91A16, 49N80, 93E20
\end{AMS}

\textit{This is work in progress.}

\section{Introduction} \label{sec:intro}

%\seb{Overall the introduction is too long... it is great for a thesis, but we need to be more concise and selective for the article. It takes a very long time to get to OUR contribution...}\note{done}
Climate change is among the greatest issues facing humanity. According to \cite{wuebbles2017climate}, we are now experiencing the warmest era in modern human history. The pertinent question in climate science is not whether the planet is warming. Instead, it is how much it will continue to do so, and what the consequences will be. Global average sea levels have increased by 8-9 inches since 1880, and are expected to continue this growth, with pessimistic projections calling for an increase in sea levels on the order of feet, not inches, by 2100 (see \cite{noaa_clim}). This would have disastrous consequences world-wide, and it is just one of the many impacts that climate change may have. Irregular weather patterns, more extreme natural disasters, and climate migration are other possible consequences.

Governments around the world have agreed to attempt to curtail their emissions in an effort to avoid or mitigate the worst of the consequences of climate change. The most notable recent example of this commitment is the Paris Agreement, signed in 2015 (see \cite{agreement2015paris}),
%\seb{would be good to put in a citation here}\note{done}
in which the world's largest polluters agreed to lowering their future carbon dioxide (CO$_2$) emissions, which is one of the primary drivers of climate change. While this agreement was generally seen as a positive direction, it lacks binding mechanisms that force nations to curtail emissions. It is up to individual countries to decide how (and indeed, whether) they reduce emissions. 

This leads to the idea of environmental policy and regulation as a means for each country, or region within a country, to  achieve their climate goals. There are several common environmental policies that have been used around the globe to reduce emissions. These include carbon taxes and cap-and-trade (C\&T) markets. More recently, there has also been focus on Renewable Energy Certificate (REC) markets, and on specific types of these markets, such as Solar Renewable Energy Certificate (SREC) markets. C\&T, REC, and SREC markets all fall into the category of market-based emissions regulation policies.

While the primary focus regarding the study of these policies is justly on their efficacy in combating climate change, we consider a related but different issue. We take for granted that there are valid environmental reasons for national and regional governments to enact policies such as REC markets. What we aim to do is to build off the work studying how firms, individually, or collectively optimise. A non-exhaustive list of works which study this either in SREC or C\&T markets is \cite{shrivats2020mean, shrivats2020optimal, carmona_coulon_schwarz_2012, hitzemann2018equilibrium, carmona_fehr_hinz_2009}.%\seb{need some citations.. ours included of course!}\note{done}

These works primarily focus on the behaviour of firms who are optimizing within these markets. We wish to consider an additional layer of complexity. Specifically, we consider how to design regulated markets in an optimal way knowing very well that firms will optimise behaviour against the regulation. Such an understanding is obviously useful to both regulated firms within these markets, and the regulators themselves.

%\note{added paragraph below to state contribution earlier in the introduction}
To do so, we formulate the regulated agents' problem through a stochastic game which we solve for an arbitrary REC market design imposed by a regulator. The regulator's problem is then defined on top of this. As a result, we formulate the joint problem as a principal agent mean field game (PA-MFG). Specifically, it is a scenario where a principal (the regulator) sets out the rules of the market in order to achieve their own goals, and the agents (regulated firms) operate amongst one another in a mean field game to navigate the imposed market in an optimal manner. Our contributions are preliminary results of optimal SREC market design that arise under this model formulation. Our methodology combines different aspects of approaches from principal agent game literature and stochastic game literature.

A full literature review is contained in Section \ref{sec:lit_review}. Before then, it is worth explicitly positioning this work relative to the extant literature, discussing the main influences of this work, and where we differ. This is followed by a brief discussion of the mechanics of REC markets.

PA-MFGs are novel in general, and the application of PA-MFG theory to environmental markets (whether REC or C\&T) is heretofore unconsidered in the literature. The MFG side of this problem was considered in \cite{shrivats2020mean}. We directly build off the regulated agent's problem formulation in the MFG side of our PA-MFG problem, using a model that is heavily inspired from that work, but also features new additions to make it more realistic. In particular, we now allow for firms to choose between short term excess REC generation and long term capital improvements to increase their baseline REC generation capacity. 

However, the more complex and challenging aspect of this problem is the interaction between the principal's problem and the mean field game which governs the regulated agents, resulting in a problem of mean field type. This is studied in detail in \cite{elie2019tale}, whose approach we are highly influenced by. However, one key difference between our work and the scenario covered by \cite{elie2019tale} is that the stochastic game which comprises MFG part of our PA-MFG includes a market clearing condition which is not included in the setup of \cite{elie2019tale}. This complexity and the problem setup generally lends the principal's problem of mean field type more naturally to the extended probabilistic approach of \cite{acciaio2019extended}. In the absence of the market clearing condition, one could choose to use the analytic (PDE) approach discussed in \cite{elie2019tale}. Additionally, we use the strong formulation of our model (and therefore any potential solution) as opposed to the weak formulation as in \cite{elie2019tale} (and other PA works such as \cite{cvitanic2018dynamic}). Like \cite{shrivats2020mean}, we consider the single period REC market. 

\subsection{The Basics of REC Markets}

REC markets belong to the class of so-called market-based emissions regulation policies. The most well-known of the policies which fall under this umbrella are carbon C\&T markets.

In carbon C\&T markets, regulators impose a limit on the amount of  $\text{CO}_2$ that regulated firms can emit during a certain time period (referred to as a compliance period). They also distribute allowances (credits) to individual firms in the amount of this limit, each allowing for a unit of $\text{CO}_2$ emission, usually one tonne. Firms must offset each of their units of emissions with an allowance, or face a monetary penalty for each allowance they are lacking. These allowances are tradable assets with prices set by supply and demand, allowing firms who require more credits than what they were allocated to buy them, and firms who require less to sell them. In this way, C\&T markets aim to find an efficient way of allocating the costs of $\text{CO}_2$ abatement across the regulated firms.

In practice, these systems regulate multiple consecutive and disjoint compliance periods, which are linked together through mechanisms such as \textit{banking}, where unused allowances in period-$n$ can be carried over to period-$(n+1)$. Other linking mechanisms include \textit{borrowing} from future periods (where a firm may reduce its allotment of allowances in period-$(n+1)$ in order to use them in period-$n$) and \textit{withdrawal}, where non-compliance in period-$n$ reduces period-$(n+1)$ allowances by the amount of non-compliance (in addition to the monetary penalty previously mentioned).

A closely related alternative to C\&T markets are REC markets. A regulator sets a floor on the amount of energy generated from renewable sources for each firm (based on a percentage of their total energy generation), and provides certificates for each MWh of energy produced via these means. This is also known as a Renewable Portfolio Standard (RPS), and typically apply to private Load Serving Entities (LSEs), also known as electricity suppliers. To ensure compliance, each firm must surrender certificates totaling the floor at the end of each compliance period, with a monetary penalty paid for each lacking certificate. The certificates are traded assets, allowing regulated LSEs to make a choice about whether to produce electricity from renewable means themselves, or purchase the certificates on the market (or a mix of both). In either case, whether purchasing or producing clean, RECs induce the production of clean energy, as all purchased RECs must have been generated through clean generation means.

%\note{removed a paragraph here that was extraneous background info to shorten up the introduction a bit}%Note that generators of electricity are not always electricity supply companies. Firms that generate renewable electricity but are not LSEs may acquire RECs, but do not have to submit them to the regulator, as they do not face a compliance requirement like LSEs do. In this work, we generally do not focus on these parties, as they are typically outnumbered in size and scale by regulated LSEs, and because their behaviour is relatively trivial, as they face no compliance requirement. 

REC markets can be used to encourage growth of a particular type of renewable energy. This is accomplished through a mechanism known as a `carve-out', requiring a specific type of renewable to be used, such as nuclear, hydroelectric, or solar. The most common of these systems are Solar REC markets (SREC markets), which have been implemented in many areas of the northeastern United States. The implementation of one as part of an RPS creates an SREC market, as opposed to a REC market. Mathematically, there is little difference between REC markets generally and SREC markets in particular. We will use the terms almost interchangeably in this work. %\seb{should we not change that... and just refer to RECs instead?}\note{it could go either way - we can use them more or less interchangeably, so it's our choice}

The similarities between carbon C\&T markets and REC markets are clear; however, there are also some notable differences. One key difference is that uncertainty in the former market is the supply of certificates (driven by generation), while in the latter, uncertainty lies in the demand for allowances (driven by an emissions process). REC markets typically allow banking but borrowing and withdrawal are not. Broadly speaking, REC markets can be considered the inverse of a C\&T system. 

\subsection{Literature Review} \label{sec:lit_review}

Next, we provide an overview of the extant literature in three categories: REC / C\&T markets, mean field games, and principal agent games.

\subsubsection{REC / C\&T Literature}
The existing literature on REC markets focuses heavily on the price formation of certificates. \cite{coulon_khazaei_powell_2015} proposes a stochastic model for economy-wide SREC generation, calibrate it to the New Jersey SREC market, and ultimately solves for the equilibrium SREC price. They further investigate the role of regulatory parameters and discuss potential takeaways for the efficient design of SREC markets. Other works, most notably \cite{amundsen2006price} and \cite{hustveit2017tradable}, discuss REC price volatility and analyze the price dynamics and the reasons for volatility. \cite{khazaei2017adapt}  go one step further, and propose an alternative SREC market design that stabilizes SREC prices. 
\cite{shrivats2020optimal} flips the typical focus of works in SREC markets, using an exogenous price process as an input to model optimal behaviour on the part of a regulated agent in an SREC market, conducting numerical simulation studies to understand and characterize the nature of the firm's optimal controls.

There is considerably more literature on carbon C\&T markets, which REC markets resemble. \cite{seifert_uhrig-homburg_wagner_2008} represent firm behaviour as the solution to an optimal control problem from the perspective of a central planner, in a single-period C\&T market, who aims to optimize total expected societal cost. The authors then characterize and solve for the carbon allowance price process. This is further extended to a multi-period C\&T market in \cite{hitzemann2018equilibrium}, which leads to the equilibrium carbon allowance price being expressed as a strip of European binary options written on economy-wide emissions. Agents' optimal strategies and properties of allowance prices are also studied by \cite{carmona2010market} and \cite{carmona_fehr_hinz_2009} via functional analysis arguments, within a single compliance period setup. Both works  make significant contributions through detailed quantitative analyses of potential shortcomings of these markets and their alternatives (in \cite{carmona2010market}) and of the certificate price and its properties (in \cite{carmona_fehr_hinz_2009}). In each of these works, the authors argue for the equivalence between the solution to the optimization problem for the central planner and an equilibrium solution whereby each individual agent optimizes their profit. %These works do not, however, incorporate trading frictions into their model. Additionally, they focus on the implications of the solution to the optimization problem posed on the certificate price, as opposed to the nature of the optimal behaviour of the regulated agents themselves, which is the primary focus of this work.
There are also notable works on structural models for financial instruments in emissions markets, such as \cite{howison_schwarz_2012} and \cite{carmona_coulon_schwarz_2012}. Finally, the problem of dynamic allocation of carbon allowances in a C\&T market is considered in \cite{aid2021optimal}.

\subsubsection{Mean Field Games Literature}

Mean field games (MFGs) are sprung from the works of \cite{huang2007large,huang2006large}, and \cite{Lasry2006a, Lasry2006b, lasry2007mean}. Broadly speaking, the methodology provides approximate Nash equilibria in stochastic differential games with symmetric interactions and a large number of players. A given player is impacted by the presence and the behavior of others through the empirical distribution of their private states. Such games are generally intractable, however, the key insight of the MFG methodology  is that taking the number of players to infinity allows  the  empirical distribution of the states of agents to be replaced by a mean field distribution. This greatly simplifies the analysis  and allows for a Nash equilibrium to be found in the infinite player limit. Just as crucially, this infinite player game can be proven to provide an approximate Nash equilibrium for the finite player game.

%\seb{Not sure we need all of this discussion on numerical methods for the paper}\note{removed} 
Many extensions and generalizations of mean field games and their applications exist. Among them include the probabilistic approach to MFGs, MFGs with common noise, and the master equation approach, see \cite{carmona2018probabilistic, carmona2013probabilistic, MasterEq2019, BensoussanBook2013}. Recently, MFGs with price formation have grown in prominence, with \cite{fujii2020mean, feron2021price, gomes2020mean2, gomes2020mean} all as contemporary examples. 

% In addition to the seminal works introducing MFGs, the numerical analysis of MFGs has also been well-studied recently. In particular, \cite{angiuli2019cemracs} and \cite{chassagneux2019numerical} study numerical approaches to the solutions of McKean-Vlasov forward-backward stochastic differential equations (MV-FBSDEs), which often characterize the solutions to MFGs. \cite{aziz2016mean} proposes a computational methodology to solve a set of forward-backward PDEs related to a MFG formulation of cellular communication networks. 

% Mean field type problems are also closely associated with MFGs. See \cite{carmona2013control} for a detailed description of the similarities and key differences between these problems, and \cite{carmona2015forward} for an analysis of the probabilistic approach to solve mean field type problems. 

MFGs have found numerous applications in engineering (\cite{aziz2016mean, KIZILKALE2019,Tembine2017}), economics (\cite{Gomes2014,Gomes2016}), and  in particular mathematical finance including optimal execution problems and portfolio trading (\cite{CarmonaLacker2015, Mojtaba2015, ThesisDena2019, FirooziISDG2017, FirooziPakniyatCainesCDC2017, Cardaliaguet2018,Lehalle2019, casgrain2018mean, Horst2018, David-Yuri2020}), systemic risk (\cite{CarmonaSysRisk2015,CarmonaSysRisk2018,JaimungalSysRisk2017}), commodities markets (\cite{aid2017coordination, Mouzouni2019,Sircar2017, brown2017oil, ludkovski2017mean}), and even cryptocurrencies (\cite{li2019mean}) -- just to name a few important contributions. 

Finally, MFGs have been applied to SREC markets through \cite{shrivats2020mean}, which we will build on heavily in this work and extend to a PA-MFG setup.

\subsubsection{Principal Agent Games Literature}

%\dena{We can add a general description of principal agent problems at the beginning followed by the literature on this topic.}
%\note{added}
Principal agent games encompass a variety of situations that involve the interaction between two distinct entities. The first entity is called a principal, who proposes a contract to the second entity, an agent. This contract is generally to perform some actions on behalf of or in conjunction with the principal. These actions require effort (cost) on the agent's part, who is typically compensated by the principal for doing so. The principal may run into the issue of moral hazard, as they typically only observe the outcome of the task, as opposed to the effort. The goal for the principal is to design a contract for the agent, which maximises the principal's own utility, but is also palatable enough to the agent for them to accept it.

\cite{holmstrom1987aggregation} is the first seminal paper on principal agent problems in continuous time. The paper considers a principal and agent with constant absolute risk aversion, where the agent's effort influences the drift of the output process, but not volatility. Here, they show the optimal contract is linear. This work was extended in various ways (see \cite{schattler1993first}, \cite{hellwig2002discrete}, among many others). More recently, \cite{sannikov2008continuous}  investigated a class of infinite-horizon principal agent games where the the principal rewards the agent continuously. Critically, this work leveraged the dynamic nature of the agent’s value function stemming from the dynamic programming principle. This allowed the principal’s problem to be formulated as a tractable optimal control
problem, and was further made rigorous and expanded upon in \cite{cvitanic2018dynamic} and \cite{cvitanic2017moral}, using second order BSDEs.

These works focused on the case where a single principal contracted a single agent. Many real-world situations, however, have a principal influencing a number of agents. With large populations of agents, finding equlibria is generally intractable, which motivates the idea of applying mean field game theory to this principal agent problem. This is exactly the approach \cite{elie2019tale} %\dena{Somewhere we should mention that our framework compared to Eli's includes a clearing condition which changes things a bit.}\note{added that note earlier on}
and \cite{carmona2018finite} take, the former in a continuous state space and the latter in a discrete state space.

Here, we use similar ideas, but take a different approach to \cite{elie2019tale} and \cite{cvitanic2018dynamic}, and analyse the principal-many-agent problem in the context of regulated REC markets. In particular, we differ from these works by marrying their general approach with the extended probabilistic approach espoused in \cite{acciaio2019extended}. This helps us deal with one of our additional complexities, a market clearing condition in the MFG. We also differ from these works by operating within the strong formulation of the problem, as opposed to the weak formulation. %\seb{We need to state what are our contributions and how we differ from these works up front -- much earlier on. }\note{added a paragraph about this earlier on}

\subsection{Outline}

We conclude this lengthy introduction with a brief outline of what is to come in this work. In \Cref{sec:model} we present the mathematical model we use to govern the REC system from both the agents' perspective and the principal's. In \Cref{sec:problem_solving}, we provide the best response of the agent, and attempt to find solutions for the principal's problem. We conclude in \Cref{sec:conclusion} with a discussion of where this work stands, and the next steps we plan on taking.
\section{Model} \label{sec:model}

We will discuss the single-period framework for REC markets, with the following rules. The market governs a compliance period from $[0,T]$, denoted as $\mfT$. A firm obtains RECs in each period, with their terminal RECs denoted by $X_T$. At time $T$ a firm must pay (or possibly is paid) $C$, which represents some generic penalty chosen by the regulator, that is known to the firms. In this work, we assume (in a slight abuse of notation) that $C = C(X_T)$. A common choice of $C$ is $P(R - X_T)_+$, where $P$, $R \in \RR$. $C$ can only be contingent on $X_T$, not on the control processes of the agents. $C$ is usually the same for all agents. We will pose our problem without assuming this is the case for increased flexibility. There are assumed to be no costs after time $T$.  \dena{in our prior work the requirement could be different from one company to another.}

Firms receive RECs through the generation of electricity via a particular energy source (depending on the market, this could be solar, nuclear, etc.). One REC typically corresponds to one MWh of electricity produced via the target energy source. A firm may also purchase or sell RECs on the market. After $T$, all firms forfeit any remaining RECs. $T$ can be thought of as `the end of the world' -- there are no costs associated with any time after this. 
%\note{Do we need to give some idea about size of penalties and REC price? depends on how we position this}
% The most notable REC markets in the real world are often SREC markets. The New Jersey SREC market was the largest and most developed in North America. It has recently been phased out (as of August 2021) during the course of the writing of this work, replaced by a system where SRECs have a fixed price, rather than one set by supply and demand). The penalty function in the NJ S $C(x) = P(R - x)_+$, with the current penalty for non-compliance as \$$258$ per unit of lacking REC, with the current REC price slightly below that, at roughly \$$230$ per REC\footnote{As of February 13, 2020} \cite{SRECTrade}. For other practical REC market details, we refer the interested reader to \cite{coulon_khazaei_powell_2015} and \cite{shrivats2020optimal}.
%\seb{should we give some idea about the size of penalties and REC price right up front? etc..?} \note{addressed}

Our goal is to extend the work of \cite{shrivats2020mean} by considering the goals of the regulator in a REC market, in addition to the firms. We use techniques from principal agent games accordingly, roughly following the blue-print laid out in \cite{elie2019tale}, though we will use the probabilistic approach of handling extended mean field type problems as in \cite{acciaio2019extended} as opposed to the analytic approach in \cite{elie2019tale}. As such, we will describe two inter-related problems, known as the agent's (firm's) problem, and the principal's (regulator's) problem. 

The principal and agents are in a leader-follower relationship. Therefore, we will present this report by discussing the agent's problem, the principal's problem, then solving the agent's problem given an arbitrary choice of $C$ from the principal, and finally, by finding the optimal choice of $C$ given the known best response of the agents' to an arbitrary $C$. In doing so, we encode the follower-leader dynamic of their relationship into the solution methodology of our problem. 

\subsection{The Agent's Problem} \label{sec:agent_problem}

We first describe the agent's problem in this generic REC market setup, with $N$ agents in total, each belonging to one of $K$ distinct sub-populations. Agents in each sub-population share the same model parameters which are different from those of agents in other sub-populations. 

We work on the filtered probability space $(\Omega, \mathcal{F}, (\Ff_t)_{t \in \mfT}, \PP)$. All processes 
%\seb{why random variables?} \note{addressed} 
are assumed to be $\Ff$-adapted
%\seb{do you mean $\Ff$-adapted/$\Ff$-predictable rather than $\Ff_t$-measurable?}
unless otherwise stated. The filtration is defined later in this section. 

We denote the set of agent types by $\mcK := \{1, ..., K\}$,  the subset of agents belonging to sub-population $k\in\mcK$ by $\mfN_k$, and the set including all agents by $\mfN$. The notation and setup here are a modification of \cite{shrivats2020mean}. 
%\dena{This sentence doesn't read well!}\note{reworded}
Here, the principal is viewed as exogenously specifying a penalty function.
%\seb{how about penalty rather than rule-set?}\note{changed}
Later, we will consider the choice the principal faces in choosing the penalty, knowing the agents' best response to an arbitrary penalty function it imposes.

Agents seek to modulate their planned excess REC generation $(g_t^i)_{t \in \mfT}$, trading $(\Gamma_t^i)_{t \in \mfT}$, and capacity expansion $(\alpha_t^i)_{t \in \mfT}$ behaviour in order to navigate the REC market in a profit-maximizing way. Planned excess REC generation rate and capacity expansion are restricted to be positive, while trading rate may be positive or negative. We denote the collection of excess generation rates, trading rates, and capacity expansion rates by $\bm{g_t} \coloneqq (g_t^1, \cdots, ,g_t^N)$, $\bm{\Gamma_t} \coloneqq (\Gamma_t^1, \cdots, \Gamma_t^N)$, and $\bm{\alpha_t} \coloneqq (\alpha_t^1, \cdots, \alpha_t^N)$ respectively.

In an arbitrary time period $[t_1, t_2]$, the firm aims to generate $G^{i}_{t_1,t_2}:=\int_{t_1}^{t_2} (h_t^k + A_t^i + g_t^i) \,dt$, where $h_t^k$ is a deterministic baseline generation level (RECS/year), and $A_t^i$ is the additional REC generating capacity above this that the firm has built through its choice of $\alpha$ ($A_0^i = 0$). The firm in fact generates $G^{i, (r)}_{t_1,t_2} := \int_{t_1}^{t_2} (h_t^k + A_t^i + g_t^i) dt + \int_{t_1}^{t_2} \sigma^k \,dW_t^i$, where $\sigma^k$ is a deterministic function of time. The diffusive term $\sigma^k dW_t^{i}$ may be interpreted as the generation rate uncertainty at $t$, reflecting the inherent uncertainty that exists in renewable energy generation. Methods similar to \cite{coulon_khazaei_powell_2015} may be used to estimate $h_t^k$ and $\sigma^k$. We assume that $0 \leq h_t^k < \infty$ for all $t$. 

Excess REC generation ($g_t^i$) can be considered a short-term transient increase in generation rate which carries with it some cost when non-zero, but does not impact long-run REC generation. For example, a firm may choose to rent solar panels from another generator in order to increase their REC generation, and would be forced to pay to do so.

Meanwhile, $A_t^i$ represents the long-term increases the firm has made to its REC generation capacity through their choice of $\alpha_t^i$, effectively increasing their baseline generation level and the rate at which the firm can generate RECs for 0 additional cost. For example, a firm can choose to invest in new solar panels or build a new solar farm, which increases the amount of RECs it can generate both now and in the future. In this work, we assume there is no delay between investing in capacity expansion and an impact to the REC producing capability of the firm. This would be a straightforward addition to the model, though adding a stochastic delay would be considerably harder.

Therefore, the firm's state variables include its controlled REC process (for a firm $i$ in sub-population $k$) and its capacity, which are described by the following SDEs:
\begin{align}
    dX_t^{i,[N]} &= (h_t^k +g_t^i + \Gamma_t^i + A_t^i)dt + \sigma^k dW_t^i, \\
    dA_t^{i, [N]} &= \alpha_t^i dt, \qquad i\in \mfN_k,
\end{align}
where $W = \{W^i = (W_t^i)_{t \in \mathfrak{T}}, i \in \mathfrak{N}\}$ is a set of N independent standard Brownian motions,
%\seb{Is $W$ meant to be the collection of all agents Brownian motions... or jsut for agent $i$... not sure about notation here} \dena{The notation is modified.} 
and $W^i$ is progressively measurable with respect to
the filtration $\mathcal{F}^W \coloneqq (\mathcal{F}_t^W)_{t\in \mfT} \subset \mathcal{F}$ generated by $W$. In a slight abuse of notation, we denote $\bm{X_t^{i, [N]}} := (X_t^{i, [N]}, A_t^{i, [N]})^T \in \RR^2$ as the vector of states for agent $i$.

In this work, firms face costs for generating RECs, costs due to trading frictions, expanding generation capacity, and costs (or potentially profits) from trading RECs. Each firm is endowed with an initial inventory of RECs, given by the collection of random variables $\{\xi^i\}_{i = 1}^N$.
%\note{Should we later on say that these are $\mathcal{G}$-measurable}
%\dena{assumptions on the independence of Brownian motions and initial states.} 
The empirical distribution of REC inventory is defined as
\begin{equation}
    \mu^{[N]}_t(dx) = \frac{1}{N} \sum_{i=1}^N \delta_{X_t^{i, [N]}}(dx), 
\end{equation}
where $\delta_y(\cdot)$ denotes the Dirac measure with unit point mass at $y$. 

Similarly, the empirical distribution of REC generation capacity is defined as

\begin{equation}
    \nu^{[N]}_t(dx) = \frac{1}{N} \sum_{i=1}^N \delta_{A_t^{i, [N]}}(dx). 
\end{equation}

%\dena{I added the empirical distribution. Havingthis an alternative for the price notation is $S^{\mu^{(N)}}_t$.}
\begin{assumption} \label{IntialStateAss}
The initial states $\{\xi^i\}_{i = 1}^N$ 
%defined on $(\Omega, \mathcal{F},(\mathcal{F}_t)_{t \in \mfT},\mathbb{P})$
are identically distributed, mutually independent, and independent of $\mathcal{F}^W$. Moreover, $\sup_{i} \mathbb{E}[\Vert \xi^i\Vert^2] \leq c < \infty $, $i \in\mfN$, with $c$ independent of $N$. For notational convenience, we say that $\xi^i \sim \mu_0^{(k)}$, for all $i \in \mfN_k$, for all $k \in \mcK$. The initial distribution of generation capacity is assumed to be a degenerate distribution with a point mass at $0$. That is, $A_0^i = 0$ for all $i \in \mfN$.  %\dena{how about the initial dist. of capacity?}\note{addressed}
\end{assumption} 

Let $\theta_t^{[N]} := \mu_t^{[N]} \times \nu_t^{[N]}$ represent the product measure, and thus the empirical bivariate distribution of the agents' states. 

The REC price is modeled endogenously; that is, its value is  derived from the supply and demand of RECs. We denote the REC price by $S^{\theta^{[N]}}_t = (S_t^{\theta^{[N]}})_{t\in\mfT}$, which emphasizes the dependence of the REC price on the controlled states of the agents. Indeed,  agents interact with one another through the REC price. As such, we must solve for the agents' optimal strategies simultaneously.

We do not make explicit assumptions about the dynamics of the REC price process, instead defining it as the price which satisfies the clearing condition:
\begin{equation}
\frac{1}{N}{\sum_{i \in \mathfrak{N}} \Gamma_t^{i, \star}} = \frac{1}{N}{\sum_{k\in\mcK} \sum_{i\in\mfN_k} \Gamma_t^{i, \star}} = 0, \label{eq:clearing_condition}
\end{equation}
$ t-a.e.,\, \PP-a.s.$, for all $t \in \mfT$.

Each agent seeks to navigate the REC market with its given rules at minimum cost. Specifically, agents attempt to minimize the following cost functional\dena{Shouldn't it be $(g_u^i-h_u)^2$? Conditioned on $\mathbf{X_0^i}$?} %\dena{$\alpha$ should be included as an argument of cost functional? also boldface $alpha$ should be defined.} \note{addressed both}
\begin{equation}
    J^{a, i}(\bm{g}, \bm{\Gamma}, \bm{\alpha}) = \EE\left[\int_0^T \left(\tfrac{\zeta^k}{2}(g_u^i)^2 + \tfrac{\gamma^k}{2}(\Gamma_u^i)^2 + \tfrac{\beta^k}{2}(\alpha_u^i)^2 +  S_u^{\theta^{[N]}} \Gamma_u^i \right)du + C^k(X_T^{i, [N]})\rvert X_0^i\right].\label{eq:agent_cost}
\end{equation}
This holds for all $i \in \mfN_k$ and for all $k \in \mcK$. This cost functional has a similar structure as in \cite{shrivats2020mean}, with the agents' objective comprising of five distinct terms. 

The first term corresponds to costs associated with short-term deviation from baseline REC generation. Specifically, the agent incurs the cost $\tfrac{\zeta^k}{2}(g_u^i)^2$ %\dena{$(g_u^i-h_u)^2$?}\note{it is correct as written; discussed on Slack} 
per unit time for its level of increase from baseline REC generation. This choice of quadratic generation cost is convenient, as it is both differentiable and convex, both of which are desirable properties for our analysis. This cost is best interpreted as a firm renting capacity to increase their planned REC generation rate, as opposed to an investment cost where they expand their long-run REC generation capacity. 

%\note{Future note: use one-way quadratic, or change the structure to adjust the interpretation from `renting' to `investing'}

The second term corresponds to a trading speed penalty, where the firm incurs a quadratic penalty in the amount traded, per unit time. This is incorporated in order to induce a constraint on their trading speed.  This cost introduces  a key difference between our model and the extant literature in the C\&T world. As mentioned in Section \ref{sec:intro}, prior works in this field do not incorporate trading frictions into their model (see \cite{hitzemann2018equilibrium}, \cite{seifert_uhrig-homburg_wagner_2008}, \cite{carmona_fehr_hinz_2009}, and \cite{carmona2010market}). 

The third term corresponds to an investment cost, which is quadratic in $\alpha$. $\alpha$ represents an increase in REC generation per year squared {\color{black} due to capacity expansion}. Put another way, it is the acceleration of total RECs {\color{black} generated due to capacity expansion}. This cost is best interpreted as a long-run improvement to REC generation capacity.

The fourth term corresponds to the cost (revenue) generated when purchasing (selling) a REC on the market, with the firm paying (receiving) the equilibrium REC price $S_t^{\theta^{[N]}}$. 

The fifth and final term corresponds to the non-compliance penalty the agent pays to the principal at the end of the compliance period. In practice, this is generally the same function for all sub-populations. However, we relax this assumption. A common choice for this is a `hockey-stick' function that imposes a linear penalty on the agents if they fail to supply a certain number of RECs (denoted $R$), and is flat afterwards, i.e., $C(X_T^{i, [N]}) = P\,(R - X_T^{i, [N]})_+$. In our model, however, the function $C^k$ is the principal's control. Agents take it as given and exogenous, thus for their problem, we keep $C^k(\cdot)$ generic. 

%\note{need to comment about the subtle differences between res cost in the two settings}

We next denote the filtration $\Gg^i= (\Gg_t^i)_{t \in \mathfrak{T}}$ that an agent adapts their strategy to and is defined as 
\begin{equation}\label{individualFilt}
    \Gg_t^i := \sigma\left((\bm{X_u^{i, [N]}})_{u \in [0, t]}\right)  \vee \sigma\left(\left(S_u^{\theta^{[N]}}\right)_{u \in [0, t]}\right).%
\end{equation}
This is the $\sigma$-algebra generated by the $i$-th firm's state variables and the REC price path. Note we assume that all firms have knowledge of the initial distribution (but not the actual value) of other firms' RECs and capacity %\dena{how about for capacity expansion?}\note{addressed}. The full filtration
$\Ff=(\Ff_t)_{t\ge0}$ is generated by the sequence of $\sigma$-algebras $\Ff_t=\bigvee_{i\in\mfN} \Gg_t^i$. 

The set of square integrable controls is defined as
\begin{equation}
    \mathbb{H}_t^2 := \left\{\left. (g, \Gamma, \alpha) : \Omega \times \mathfrak{T} \rightarrow \R^3 \; \right|\; \EE\left[{\textstyle\int_0^T} \left( (g_t)^2 + (\Gamma_t)^2 + (\alpha_t)^2 \right) dt\right] < \infty\right\}.
\end{equation}

\begin{assumption} \label{ass: MinorContrAction} The set of admissible controls for firm $i\in\mfN$  is %\dena{$\alpha$ is not included in the square integrability and measurability condition.}\note{addressed}
\begin{align}
    \mcA^i := \left\{ (g, \Gamma, \alpha) \in \mathbb{H}_t^2\,\, \text{s.t.} \,\, g_t \geq 0\,, \alpha_t \geq 0\, \text{ for all } \, t \in \mathfrak{T}\,\, \text{and} \,\, (g, \Gamma) \text { $\Gg^i$-adapted}\,\, \right\}. \label{firmsAdmissibleCntrl}
\end{align}
\end{assumption}
This is a closed and convex set. As an individual firm cannot observe another firms' REC inventories, the restriction in the set above is to decentralize $\Gg^i$-adapted strategies.

As each agent is impacted by the distribution of states across all agents, we must solve for their optimal strategies collectively. Therefore, the agent's problem can be expressed as aiming to find the set of strategies that form a Nash equilibrium. That is, we search for a collection of controls %\dena{$\alpha$ is not included here and in the equations below.}\note{addressed}
$\{(g^{i,\star}, \Gamma^{i,\star}, \alpha^{i, \star}) \in \mcA^i \}_{i\in\N}$ such that 

\begin{equation}
(g^{i,\star}, \Gamma^{i,\star}, \alpha^{i, \star}) = \underset{(g, \Gamma, \alpha) \in \mcA^i}{\arg \inf} \, J^{a,i}(g^i, \Gamma^i, \alpha^i, \bm{g}^{-i, \star}, \bm{\Gamma}^{-i, \star}, \bm{\alpha}^{-i, \star}), \hspace{5mm} \forall i \in \mfN_k, \, \forall k \in \mcK, \label{eq:problem_statement}
\end{equation}
where \dena{$\alpha$ should be included.}\note{addressed} $(g^i, \Gamma^i, \alpha^i, \bm{g}^{-i, \star}, \bm{\Gamma}^{-i, \star}, \bm{\alpha}^{-i, \star})$ denotes the set of strategies $(\bm{g}^\star, \bm{\Gamma}^\star, \bm{\alpha}^\star)$ with $(g^{i, \star}, \Gamma^{i, \star}, \alpha^{i, \star})$  replaced by $(g^i, \Gamma^i, \alpha^i)$. This means that an individual agent cannot benefit by unilaterally deviating from the Nash strategy. 

In the finite-agent case, this is intractable due to the high dimensionality of the problem. %\dena{also due to dimensionality.}\note{added} 
Instead, we consider the mean field limit of this problem, with $N \rightarrow \infty$. In accordance with this shift in the problem, we make a few assumptions and introduce some new notation.%\dena{It is introducing new notation rather than notational change.}\note{adjusted}

\textcolor{black}{We maintain our notation of indexing by firms, even in the infinite limit case. This allows us to be precise about what we are referring to, with no loss of clarity, as we will explicitly point out the sub-population each agent belongs to. Where appropriate, however, we use a superscript ${\cdot}^\ikrep$ rather than a superscript ${\cdot}^i$ to denote quantities relating to a representative (unspecified index) agent belonging to sub-population $k$. Parameters such as $\zeta, \gamma, \beta, h$ retain their superscript $k$ notation without parentheses, as they are real numbers and not a representative process.}

Accordingly, we denote the $i$-th agent's REC inventory in the infinite-population limit by $X^i = (X_t^i)_{t \in \mfT}$, and their generation capacity by $(A_t^i)_{t \in \mfT}$. Consistent with the previous paragraph, we often additionally consider the inventory of a representative agent from sub-population $k$, whose REC inventory and generation capacity we denote by $X_t^\ikrep$ and $A_t^\ikrep$ respectively, for convenience. This notation will primarily arise when we consider the principal's problem. 

The dynamics for the $i$-th agent's (belonging to sub-population $k$) REC inventory and REC generation capacity level are now
\begin{align}
dX_t^i &= (h_t^k + g_t^i + \Gamma_t^i + A_t^i)\, dt + \sigma_t^k \,dW_t^i, \label{eq:state_dynamicsMF}\\
dA_t^i &= \alpha_t^i dt \hspace{5mm}\text{ $\forall i \in \mfN_k$}\,, \label{eq:capacity_dynamicsMF}
\end{align}
where $W = \{W^i = (W_t^i)_{t \in \mathfrak{T}}, i\in \mathfrak{N}\}$ is a set of independent standard Brownian motions. Once again, each firm has an initial inventory of  RECs, given by the collection of random variables $\{\xi^i\}_{i = 1}^\infty$, and $A_0^i = 0$ for all $i \in \mfN$. Abusing notation slightly, we let $\bm{X_t^i} \coloneqq (X_t^i, A_t^i)^T$.

In this limit, individual agents no longer have a direct impact on one another. Instead, they interact with the mean field distribution of states, which makes the problem significantly simpler to solve. 

\begin{assumption}\label{ass:proportion}
The proportion of the total population of agents belonging to each sub-population $k\in\mcK$ converges to a constant as the number of firms ($N$) increases. That is,
\begin{equation}
\lim_{N \rightarrow \infty} \tfrac{N_k}{N} = \pi_k \in (0, 1) \quad  \PP.a.s., \; \forall k\in\mcK \label{eq:proportion}
\end{equation} 
\end{assumption}

We now define the notation used for the distribution of $X_t$ across agents.

\begin{definition} [Mean Field Distribution of REC Inventory] \label{def:mean_field_distribution_X} 
In the infinite-player setting, we denote the mean field distribution of REC inventory for agents in sub-population $ k \in \mcK$ by $\mu_t^{(k)}$. Specifically, we introduce the flow of measures
\begin{equation}
    \mu^{(k)} = (\mu_t^{(k)})_{t \in\mfT},
    \qquad \mu_t^{(k)} \in \mcP(\R), \;\forall t\in\mfT
\end{equation}
for all $k \in \mcK$,  where $\mathcal{P}(\R)$ represents the space of probability measures on $\R$, such that $\mu_t^{(k)}(A)$ is the probability that a representative agent from 
sub-population $k$ has an REC inventory belonging to the set $A \in \mathcal{B}(\R)$, at time $t$. 
%\dena{Shouldn't this be the integral of $\mu^k(dx)$ over $A$?}\seb{I think it's standard notation that $\mu(A)$ is the measure of set $A$}
Furthermore, we define%\note{refine this definition, possibly}
\begin{equation}
    \bm{\mu} = (\{\mu_t^{(k)}\}_{k \in \mcK})_{t \in \mfT}
\end{equation}
to be the flow of the collection of all mean field measures. 
\end{definition}

Similarly, we define the notation used for the distribution of $A_t$ across agents. 

\begin{definition} [Mean Field Distribution of REC Generation Capacity] \label{def:mean_field_distribution_A} 
In the infinite-player setting, we denote the mean field distribution of REC generation capacity for agents in sub-population $ k \in \mcK$ by $\nu_t^{(k)}$. Specifically, we introduce the flow of measures
\begin{equation}
    \nu^{(k)} = (\nu_t^{(k)})_{t \in\mfT},
    \qquad \nu_t^{(k)} \in \mcP(\R), \;\forall t\in\mfT
\end{equation}
for all $k \in \mcK$,  where $\mathcal{P}(\R)$ represents the space of probability measures on $\R$, such that $\nu_t^{(k)}(A)$ is the probability that a representative agent from 
sub-population $k$ has an REC generation capacity belonging to the set $A \in \mathcal{B}(\R)$, at time $t$. 
Furthermore, we define
\begin{equation}
    \bm{\nu} = (\{\nu_t^{(k)}\}_{k \in \mcK})_{t \in \mfT}
\end{equation}
to be the flow of the collection of all mean field measures. 
\end{definition}

Finally, we define the mean field distribution of states.

\begin{definition} [Mean Field Distribution of States] \label{def:mean_field_distribution} 
In the infinite-player setting, we denote the mean field distribution of the agents' states in sub-population $ k \in \mcK$ by $\theta_t^{(k)}:= \mu_t^{(k)} \times \nu_t^{(k)}$. Furthermore, we define
\begin{equation}
    \bm{\theta} := \bm{\mu} \times \bm{\nu}.
\end{equation}
\end{definition}

Observe that in \eqref{eq:agent_cost}, only the REC price is dependent on the actions of other agents. 
Accordingly, in the infinite population limit, we denote the REC price as $S_t^{\bm{\theta}}$. As all individuals are minor-agents, any single agent does not impact the price, however, the price is impacted by the mean field distribution. We assume that it is impacted only by the mean field distribution and not the mean field of controls (which would render this problem an extended MFG problem). In \cite{shrivats2020mean}, this assumption is validated (see (4.35) in the cited work). Therefore, the dependence of any agent on the mean field distribution of states arises through their dependence on the REC price process alone.

%\note{in principle, this \textbf{could} be impacted by individual agents, right? is there any possible circularity in us making this assumption and then achieving a form of S that validates it}
In the MFG, we modify the clearing condition to endogenously define the REC price to be
\begin{equation}
    \lim_{N \rightarrow \infty} \tfrac{1}{N}\sum_{i\in\mfN}\Gamma_t^i = 0,
\end{equation}
$ t-a.e.,\, \PP-a.s.$, for all $t \in \mfT$. %\note{does it make sense to have $t$ twice here?} 
This is analogous to the average trading rate (across agents)  vanishing at all times in the limit.

%\note{is the section below out of place?}
We now discuss the mechanics of the principal's choice of the penalty function. While we are still dealing with the collective problem of the agents, there must be some restrictions on the what the principal may do in order for the agents' problem to be well-posed and a logical translation of the real-world problem we are aiming to solve. 

Consider a generic agent in sub-population $k$ with inventory $X_T^\ikrep$. We observe that $C^k(X_T^{\ikrep})$ 
is a random variable, which will have a different distribution for agents in different sub-populations. We also note that the principal can choose a different non-compliance penalty for each sub-population, if they wish. To make this clear, we define
\begin{equation}
    \bm{C} := (C^k)_{k \in \mcK} := (C^k(X_T^{\ikrep}))_{k \in\mcK} \in \RR^{K}.
\end{equation}
The principal chooses over the random vector $\bm{C}$, which is equivalent to setting the $K$ functions $\ovc := (C^1, \cdots C^K)$. These are the contracts the principal enters into with the agents.

We require that the principal chooses $\bm{C}$ such that a solution to the MFG faced by the agents has a solution. Sufficient conditions for this is for $C^k$ to be non-increasing in $X_T^k$, convex, and once continuously \note{added ctsly diff here} differentiable everywhere. Note that by Alexandrov's theorem, this implies twice differentiability almost everywhere.  \textcolor{black}{We further assume the first derivative is bounded and Lipschitz, which tells us that the second derivative is bounded almost everywhere}.\note{Ensure this condition is required} Convexity of $C^k$ maintains the convexity of the optimization problem, and differentiability ensures our cost functional is G\^ateaux differentiable everywhere. While the classical `hockey stick' function previously described does not fit into this umbrella of functions, we can easily regularize it such that it does, and we could recover the regularized version as the optimal choice of $\bm{C}$ for the principal.

Subsequently, in the infinite-population limit, we denote firm-$i$'s cost functional  by $\overline{J^{i}}$. Agent $i$ thus attempts to minimize the following cost functional %\dena{$\alpha$ should be included in the argument. Running cost should include $(g^i_t-h_t)^2$?}\note{addressed}
\begin{equation}
    \overline{J}^{a, i}(g^i, \Gamma^i, \alpha^i; \bm{\theta}) = \EE\left[\int_0^T \left(\tfrac{\zeta^k}{2}(g_u^i)^2 + \tfrac{\gamma^k}{2}(\Gamma_u^i)^2 + \tfrac{\beta^k}{2}(\alpha_u^i)^2 + S_u^{\bm{\theta}} \Gamma_u^i \right)du + C^k(X_T^i)\rvert \Gg^i_0 \right],
    \; i\in\mfN_k.
    \label{eq:rec_pa_theory_agent_pc}
\end{equation}

In the infinite-player game, agent-$i$ adapts their strategy to the filtration $\Gg^i= (\Gg_t^i)_{t \in \mathfrak{T}}$, where 
\begin{equation}\label{eq:individualFilt_infpop}
    \Gg_t^i := \sigma\left((\bm{X_u^i})_{u \in [0, t]}\right) \vee \sigma\left(\left(S_u^{\bm{\theta}}\right)_{u \in [0, t]}\right)\,. 
\end{equation}
This is the infinite-player analogue of \eqref{individualFilt}. The admissible set $\mcA^i$ retains its definition in \eqref{firmsAdmissibleCntrl} but with the filtration $\Gg$ above.

We seek a Nash equilibrium for this problem. Specifically, we search for a collection of controls %\dena{$\alpha$ should be included.}\note{addressed}
$\{(g^{i,\star}, \Gamma^{i,\star}, \alpha^{i, \star}) \in \mcA^i \}_{i\in\N}$ such that
\begin{equation}
(g^{i,\star}, \Gamma^{i,\star}, \alpha^{i, \star}) = \underset{(g, \Gamma) \in \mcA^i}{\arg \inf} \, \overline{J}^{a,i}(g,\Gamma, \alpha; \bm{\mu}), \hspace{5mm} \forall i \in \N\,.  \label{eq:agent_problem_statementMF}
\end{equation}
% Similarly, let $V_t(C) = (V_t^{(1)}, \cdots, V_t^{(K)})^T$.

The solution to \eqref{eq:agent_problem_statementMF} is a fixed point problem across the space of probability measures, combined with a standard stochastic control problem. Recall that in the MFG framework,  $\bm{\theta}$ is both an \textit{input} and an \textit{output} to the problem. The agents' cost functional depends on it (through the REC price), but it must also coincide with distribution of the states through \eqref{eq:state_dynamicsMF}, which are an output of the problem, determined by the optimal controls. This results in a fixed-point problem on the space of measure flows, as is standard in MFG problems.

To elaborate, we initially treat $\bm{\theta}$ as an exogenous input to the problem. For this $\bm{\theta}$ held fixed, \eqref{eq:agent_problem_statementMF} is a standard stochastic control problem. We solve this control problem and ensure that the optimal controls result in a controlled state with a distribution that coincides exactly with $\bm{\theta}$. The vast majority of this aspect of the work is taken care of in \cite{shrivats2020mean}; here, we invoke the machinery developed there and then move on to the principal agent version of the problem.

Additionally, for a given $C^k$, let \dena{Using $A$ as superscript could be confusing as the same notation has been used for capacity expansion.\\ The set over which we are infimizing can be specified. }\note{Agreed; we can discuss this for better notations; changed upper case A to lower case} $V^{a, (k)}(C^k) = \inf \overline{J}^{a, \cdot}(g, \Gamma, \alpha;\bm{\theta})$%\seb{This looks as if all agents in group $k$ have the exact same initial optimal value... }\note{updated}
, defined for all $k \in \mcK$. This is the optimally controlled value function for a representative agent in sub-population $k$. This exists for all $k \in \mcK$.  Additionally, define %\dena{$\alpha$ should be included.}\note{addressed}
$\mcM(\ovc):=\{g^{i,\star},\Gamma^{i,\star}, \alpha^{i, \star}\}_{i\in \mfN_k, k \in \mcK}$ %\dena{so far we had index $i$ for control processes.}\note{adjusted} 
to be the collective optimal response across all agents to contracts $\ovc$. That is, $\mcM(\ovc)$ is the set of controls that satisfy \eqref{eq:agent_problem_statementMF}, along with the consistency condition that results in a solution to the MFG.

\subsection{The Principal's Problem} \label{sec:principal_problem}

The principal is the other side of the REC market.
They are the leader who sets the choice of the penalty functions $\ovc$, after which the agents respond with their optimal controls. We restrict ourselves to the penalty functions $\ovc$ such that solutions to \eqref{eq:agent_problem_statementMF} exist. That is,  we want choose $\ovc$ such that $\mcM(\ovc) \neq \emptyset$.

We prevent the principal from having unilateral power in choosing the contract, as that is both mathematically uninteresting and practically unrealistic. In practice, regulated LSEs lobby politicians and legislative bodies to ensure their needs and viability as a business are not ignored in the design of the REC markets. This feature is imposed by  a reservation cost $R_0$, which corresponds to the maximum expected cost across agents within a sub-population that the principal may impose. Any REC market that results in an expected cost to the average agent (in any of the $K$ sub-populations) above $R_0$ is deemed unviable for the regulated agents to enter into, and therefore will not be instituted, so as to not construct a REC market that runs the agents out of business. Note that this does not vary across sub-populations in our setup, though it could, and would not change the analysis significantly. 

More precisely, we choose contracts $\bm{C}$ such that the agents' best response results in $\int \inf \overline{J}^{a, \cdot}(g^{\cdot, \star}, \Gamma^{\cdot, \star})d\mu_0^{(k)} \leq R_0$ for all $k \in \mcK$, where $R_0$ represents some reservation cost.
%where it is untenable for the principal to impose a penalty function that results in any agent incurring a cost higher than $R_0$ (in expectation). 
We denote the set of admissible contracts $\bm{C}$ by%\note{Might need to add bounded first derivative here... lets us search over bounded $Y$, which might be important for demonstrating optimality}
% \begin{equation}
%     \Xi := \{\bm{C} \in \RR^{K}: \bm{C} = (C(X_T^k))_{k \in \mcK},\; C \in \mcQ, \; V^{a, k}(C) \leq R_0, \; \forall k \in \mcK\}
% \end{equation} 
\begin{align}
    \Xi := \{ &\bm{C}: \Omega \rightarrow \RR^K:\, \bm{C} = (C^k(X_T^\ikrep))_{k \in \mcK}, \, C^k \text{ convex},\, \, \frac{dC^k}{dx} \text {exists}, \,\frac{d^2C^k}{dx^2} \text { bounded a.e},  \nonumber \\ \, & C^k \text{ Lipschitz}, \, \EE[C^k(X_T^\ikrep)^2] < \infty, \int V^{a,\ikrep}(C) d\theta_0^\ikrep \leq R_0, \, \forall k \in \mcK \} \label{eq:ncc_fun_set_theory}
\end{align}
{\color{black}
\begin{align}
    \Xi := \{ &\bm{C}: \Omega \rightarrow \RR^K:\, \bm{C} = (C^k(X_T^\ikrep))_{k \in \mcK}, \, C^k \text{ convex, continuously diff.},\, \frac{d^2C^k}{dx^2} \text { exists a.e}\,  \nonumber \\ \, & \frac{dC^k}{dx} \text{ bounded and Lipschitz}, \, \EE[C^k(X_T^\ikrep)^2] < \infty, \int V^{a,\ikrep}(C) d\theta_0^\ikrep \leq R_0, \, \forall k \in \mcK \} %\label{eq:ncc_fun_set_theory}
\end{align}} \note{dont need second derivative bounded a.e. (explicitly) because Lipschitz first derivative implies the second derivative is bounded whenever it exists}

%\dena{dependence on $x_0$ in the integral.}
%\note{added Lipschitz - this might not be necessary if we can find results on the value function being lipschitz. this also impacts the def'n of $\hat{\Xi}$, \Cref{thrm:mfg_soln}, and \cref{prop:reduction}}
\textcolor{black}{In words, we are restricting the principal to choose from convex penalty functions with continuous, bounded, and Lipschitz first derivative, and a second derivative that exists almost everywhere. The Lipschitzness of the first derivative gives us that the second derivative is bounded almost everywhere. We also require that they choose a penalty function that results in the agents' MFG having a solution. These are standard assumptions, and are also adopted  in \cite{elie2019tale}. }

In the context of the principal's problem, as introduced earlier, we use the superscript notation ${\cdot}^\ikrep$ to indicate a process related to a representative agent from sub-population $k$. The principal aims to solve the following optimization problem:%\dena{it makes sense to have $\lambda$ vary with the type!}\note{changed}
\begin{align}
    V^P := \inf_{\bm{C} \in \Xi} J^P(\bm{C}), \qquad J^P(\bm{C}) := \EE\left[U_P\left(\sum_{k \in \mcK} \pi_k \left[ -\, C^k(X_T^\ikrep) - \lambda^k X_T^\ikrep \right] \right)\right], \label{eq:principal_cost}
\end{align}
%where $g_t, \Gamma_t$ represent\seb{these do not appear above... so I presume this sentence can be removed} the behaviours of the agents when faced with the penalty function $C$, and $X_T$ is the resulting controlled process. 
Here, $U_P$ represents the principal's utility function. This must be non-decreasing and convex, \dena{non-decreasing?}\note{CHECK THIS, see \href{https://web.stanford.edu/~boyd/cvxbook/bv_cvxbook.pdf}{here}}with a derivative that grows at most linearly. \note{check these requirements against what we need from AB-VC (2019)}This problem is thus similar to a Stackelberg (leader-follower) equilibrium between the principal and the infinitely many agents, who are themselves seeking a Nash equilibrium, subject to clearing conditions.

The intuition of the form of \eqref{eq:principal_cost} is to reflect that the principal aims to promote the generation of RECs (and hence wants terminal RECs to be large) while also desiring revenue in the form of non-compliance penalties to be paid to them.

As currently expressed, the principal's problem is not a standard stochastic control problem as the optimization is over some restricted random vector $\bm{C}$, as opposed to a stochastic process. We follow the steps of \cite{cvitanic2018dynamic} and \cite{elie2019tale} in order to reduce this problem to one that we can solve using existing techniques. Structurally, we will re-cast our problem to something similar to what is seen in \cite{elie2019tale}, as it will reduce to an optimal control problem that incorporates extended McKean-Vlasov dynamics.

We also note a subtlety with respect to the mean field distribution in the principal's problem which requires slightly different treatment and notation compared to the mean field distribution in the agent's problem, which is elaborated on in the remark below.

\begin{remark}
As discussed in \Cref{sec:agent_problem}, the solution methodology for the agent's problem (and more generally, classical MFG problems) keeps the mean field distribution $\bm{\theta}$ fixed initially. We perform the optimization for this fixed $\bm{\theta}$, which will imply a controlled distribution of $X_t^k$ that is dependent on this choice of $\bm{\theta}$. We then find a fixed point of $\bm{\theta}$; that is, a $\bm{\theta}$ such that the marginal distribution of the implied state $X_t^k$ is exactly %\dena{$\mu_t^{(k)}$}\note{modified}
$\mu_t^{(k)}$at all times $t \in \mfT$. This occurs because the search for a Nash equilibrium implies that an optimizing agent is holding all other agents' behaviour fixed, and under this assumption, is optimizing their own behaviour. In the infinite-population limit (when $N \rightarrow \infty$), this agent perturbing their behaviour does not make an impact on the overall distribution, and thus it is justified to hold the mean field distribution fixed, optimize for a representative agent, and then match the distribution used as an input with the implied distribution arising from the optimal controls.

By contrast, the nature of the principal's problem is such that any modification of the principal's control will impact (and perturb) all agents' behaviours. Therefore, a perturbation of the principal's control will have a significant impact on the mean field distribution, and hence it is not justified to hold it fixed while performing the principal's optimization. Instead, one must match the mean field distribution to the law of the state or control before optimizing. Intuitively, this difference reflects the fact that the principal behaves like a major agent in some sense, and has actions that impact all agents directly. These are often referred to as mean field control problems, problems of control of McKean-Vlasov dynamics or mean field type problems. We have a particular case of extended McKean-Vlasov dynamics. For a further discussion of these nuances, we refer the interested reader to \cite{carmona2013control} for problems of McKean-Vlasov dynamics, and \cite{acciaio2019extended} for problems of extended McKean-Vlasov dynamics.
\end{remark}

\section{Solution to optimization problems} \label{sec:problem_solving}

In this section, we solve for the agent's problem given a valid penalty function $\bm{C} \in \Xi$  from the principal. This will produce a best response of the agent, which in turn induces a controlled state and a mean field distribution.

We will then turn to the principal's problem armed with these, and find the $C$ that induces a response and mean field distribution that is optimal for the principal. This will be done by reducing the problem to one of control of extended stochastic dynamical systems of McKean-Vlasov type, as in \cite{acciaio2019extended}.

In this manner, we will have characterized the Stackelberg equilibrium (leader-follower) that exists between the principal and the agents.

\subsection{The best response of the agent} \label{sec:agent_response}
Here, we assume the principal chooses an arbitrary random vector $\bm{C} \in \Xi$. Due to the definition of $\Xi$ in \eqref{eq:ncc_fun_set_theory}, this guarantees the existence of a solution to the agent's problem. In particular, this is due to the assumed convexity and conditions on derivatives of the non-compliance penalty function, which make the problem similar to \cite{shrivats2020mean}. We summarize the results that we obtain from invoking the machinery developed in \cite{shrivats2020mean} in the theorem below.\dena{infinite population clearing cond is not indicated anywhere.}

\begin{theorem}[Summary of Solution of Agent's Problem] \label{thrm:rec_pa_theory_mfg_soln}
Fix $\bm{C} \in \Xi$. This induces a collection of penalty functions $\ovc$. For a firm $i$ belonging to sub-population $k$, the optimal controls of the firm are (given an exogenous mean field distribution $\bm{\theta}$):
\begin{align}
    g_t^{i, \star}  &= - \tfrac{Y_t^{X, i}}{\zeta^k} , \quad 
    \label{eq:rec_pa_theory_optG} \\
\Gamma_t^{i, \star} &=   \tfrac{1}{\gamma^k} \left( - Y_t^{X, i} - S_t^{\bm{\theta}}\right) \;\text{and} \label{eq:rec_pa_theory_optGamma} \\
\alpha_t^{i, \star} &= \tfrac{- Y_t^{A,i}}{\beta^k}. \label{eq:rec_pa_theory_optAlpha}  
\end{align}

The equilibrium REC price is then given by% \dena{we switch from $S^{\theta}$ to $s^{\mu}$}\note{addressed}
\begin{equation}
S_t^{\bm{\theta}} = - \frac{\sum_{k\in\mcK} \eta_k \,\EE[Y_t^{X, \ikrep}]}{\eta}, \,\text{ a.s.} \label{eq:rec_pa_theory_equilibrium_price}
\end{equation}
where 
\begin{equation}\label{eq:rec_pa_theory_etak}
\eta_k= \frac{\pi_k}{\gamma^k}, \;\; \eta = \sum_{k \in \mcK} \frac{\pi_k}{\gamma^k}, \,\,\,k\in\mcK.
\end{equation}

To fully characterize these optimal controls and quantities, we must specify the mean field distribution %$\bm{\theta}$ \dena{we only mention $\mu$ here.}\note{addressed} 
as well as $Y_t$. The latter is done (again, given an exogenous $\bm{\mu}$) through the following FBSDE, for a firm $i$ in sub-population $k$: \dena{can this be written in a way that eq Nos appear in the same line?}%\dena{in the FBSDE a common penalty $C$ is used.}\note{addressed}

\begin{subequations}
\begin{align}
dX_t^{i} &= \left(h_t^k - \upsilon^k\,  Y_t^{X, i}
+ \tfrac{1}{\gamma^k \eta} \sum_{j\in\mcK} \eta^j \EE[Y_t^{X, (j)}] + A_t^i\right) \,dt + \sigma_t^k dW_t^i, \qquad &X_0^i=\xi^i
\label{eq:rec_pa_theory_HetFBSDE_fwdX} 
\\
dA_t^i &= \frac{-Y_t^{A, i}}{\beta^k} dt, \label{eq:rec_pa_theory_HetFBSDE_fwdA} \qquad &A_0^i = 0 \\
dY_t^{X, i} &= Z_t^{X, i} \,dW_t^i, \qquad &Y_T^{X,i} = \frac{dC^k}{dx}(X_T^i) \label{eq:rec_pa_theory_HetFBSDE_bwdX}
\\
dY_t^{A, i} &= -Y_t^{X,i} + Z_t^{A, i} \,dW_t^i,  \label{eq:rec_pa_theory_HetFBSDE_bwdA} \qquad &Y_T^{A, i} = 0
\end{align}
\label{eqn:rec_pa_theory_FBSDE-full}%
\end{subequations}%
where
\begin{equation}
\upsilon^k := \frac{1}{\gamma^k} + \frac{1}{\zeta^k},
\end{equation}
and $\xi^i\sim \mu^{(k)}_0$. 

Furthermore, there exists a mean field distribution $\bm{\theta}$ and a progressively measurable triple $(\bm{X^i, Y^i, Z^i}) = (\bm{X_t^i, Y_t^i, Z_t^i})_{t \in \mathfrak{T}}$ that satisfy \eqref{eqn:rec_pa_theory_FBSDE-full}, such that $\theta_t^{(k)}$ coincides with $\mathcal{L}(X_t^i) \times \mathcal{L}(A_t^i)$ for all $i \in \mfN_k$, for all $k \in \mcK$. Note $\bm{Y_t^i} := (Y_t^{X,i}, Y_t^{A,i})^T$ (similar to the definition of $\bm{X}$ below \eqref{eq:state_dynamicsMF}-\eqref{eq:capacity_dynamicsMF}, with $\bm{Z}$ defined similarly. 

In particular, this choice of $\bm{\theta}$ and $(\bm{X^i, Y^i, Z^i})$  characterizes a mean field distribution that implies optimal controls that result in the law of the controlled state being exactly that of the mean field distribution we started with. $Y_t^{X, i}$ also has a Markov form $Y^k(t, X_t^{i}, A_t^{i})$ \dena{with this presentation, is it right to use $S^{\mu}$?}, which is Lipschitz in both arguments. This specification leads to a Nash equilibrium across agents, thus solving the agent's problem in its entirety. Moreover, this solution is unique. 
\end{theorem}

\begin{proof}

In this proof, we make frequent reference to \cite{shrivats2020mean}. While this problem is not identical to the problem considered in that work, the structures are extremely similar and as such, the proofs contained therein can be trivially adapted for our needs. 

The definition of $\Xi$ includes only convex and continuously differentiable functions%\dena{differentiable everywhere?}\note{corrected}
.The former, along with the similar running costs to the problem considered in \cite{shrivats2020mean} assures us that the agents' cost functional is convex (see Proposition 4.2 in \cite{shrivats2020mean}). The latter ensures the cost functional \eqref{eq:rec_pa_theory_agent_pc} is G\^ateaux differentiable (see Proposition 4.3 in \cite{shrivats2020mean}). 

This implies that an admissible control that makes the G\^ateaux derivative vanish (in all directions) is a minimizer of \eqref{eq:rec_pa_theory_agent_pc}. Moreover, convexity assures us that this minimizer is unique. 

By finding the controls for which the G\^ateaux derivative vanishes as in Propositions 4.4-4.5 %\dena{shouldn't be proposition 4 only?}\note{I'm comparing to our published version} 
of \cite{shrivats2020mean}, we therefore know that the optimal controls have the forms specified by \eqref{eq:rec_pa_theory_optG}-\eqref{eq:rec_pa_theory_optAlpha}. 
\textcolor{black}{\eqref{eq:rec_pa_theory_equilibrium_price} occurs through market clearing conditions (see Proposition 4.6 in \cite{shrivats2020mean}%\dena{shouldn't be proposition 5?}\note{addressed} in \cite{shrivats2020mean})%\dena{what about the backward part?}\note{addressed}
\dena{We don't give the definition of adjoint processes in this paper!}\note{See \eqref{eqn:rec_pa_theory_FBSDE-full}}. The FBSDE \eqref{eqn:rec_pa_theory_FBSDE-full} arises through the substitution of the optimal controls into the forward state equation (see Corollary 4.10 in \cite{shrivats2020mean}), and the definition of the adjoint process $\bm{Y_t}$. The boundedness and Lipschitz-ness %\dena{differentiability instead of continuity?}\note{I remember that I put continuity for a reason, but I can't remember what that reason was...}
of $\frac{dC^k}{dx}$ ensures that we can apply Remark 4.12  \dena{according to Delarue 2002, it seems that we need lipschitz continuity and boundedness condition for the derivative of C wrt X.}\note{This is now included in $\Xi$} in \cite{shrivats2020mean} to guarantee the Markov and Lipschitz nature of the adjoint process $Y_t^{X, i}$. }

Finally, we can apply the same proof of existence \dena{We need continuous differentiability of C for the existence (A4 in Carmona \& Delarue 2013)}\note{added throughout, and reflected in equivalence proof} and uniqueness to the MV-FBSDE \eqref{eqn:rec_pa_theory_FBSDE-full} as in Proposition 4.13 and 4.14 in \cite{shrivats2020mean}.  %\squarefill{}
\end{proof}

\begin{remark}
\color{black}
The equilibrium REC price when firms are behaving optimally can be equivalently stated in a variety of ways. That is, \dena{switching to $S^{\theta}$ again.}

\begin{align*}
S_t^{\bm{\theta}} &= - \frac{\sum_{k\in\mcK} \eta_k \,\EE[Y_t^{X, \ikrep}]}{\eta}, \,\text{ a.s.} \\
&= - \frac{\sum_{k\in\mcK} \eta_k \,\int Y^k(t, X_t^\ikrep, A_t^\ikrep) d\theta_t^\ikrep}{\eta}, \,\text{ a.s.} \\
&= -\frac{\sum_{k\in\mcK} \eta_k \,\int Y_t^{X, \ikrep} d\PP^{\bm{Y_t}}}{\eta}, \,\text{ a.s.} \\
&= S_t^{\PP^{\bm{Y_t}}}
\end{align*}

We refer to the third line as the REC price in McKean-Vlasov form, and the second line as the REC price in Markov form. 
\end{remark}

\note{The above remark is important; we need to be able to argue that these are equivalent and we can swap between them for the reduction to work}

\subsection{Recasting the Principal's Problem}
\note{Is there any doubt over whether our proofs of existence / uniqueness transfer over to the agent's problem with expansion}
Having justified the use of the methodology in \cite{shrivats2020mean} to find a Nash equilibrium among the agents for an arbitrary admissible penalty function, we must now optimize across the choice of said functions for the principal. 

To do so, we follow a similar strategy to \cite{cvitanic2018dynamic} and \cite{elie2019tale} in order to transform the principal's problem into a more standard stochastic control problem. When considering the principal's problem, we first note that the optimally controlled state process for a firm belonging to sub-population $k$ is:
%\dena{For myself: check if in Eli's setup the output process is common between all agents. What does their cost look like? As I remember the sum of states doesn't appear in the cost but the state of representative agent. }

\begin{equation}
dX_t^\ikrep = \left(h_t^k - \upsilon^k\,  Y_t^{X,\ikrep}
+ \tfrac{1}{\gamma^k \eta} \sum_{j\in\mcK} \eta^j \EE[Y_t^{X,(j)}] + A_t^\ikrep\right) \,dt + \sigma_t^k dW_t^k. \end{equation}

Note that the superscript $i$ is dropped (as it was in \Cref{sec:principal_problem}), to reflect that we are now discussing representative agents from sub-populations, as opposed to individuals. From \eqref{eq:rec_pa_theory_HetFBSDE_bwdX}, we can also see that $Y_t^{X,\ikrep}$ is a martingale for all $k \in \mcK$, and in particular, that $Y_t^{X,\ikrep} = \EE_t[Y_T^{X,\ikrep}] = \EE_t[\frac{dC^k}{dx}(X_T^\ikrep)]$. %\dena{common penalty here.}\note{addressed}

From this, it is clear to see that the principal's choice of $\ovc$ impacts $Y^{X, \ikrep}$ for each sub-population, which in turn impacts $Y^{A, \ikrep}$, as well as the agent's optimal response and controlled states. This motivates the idea of treating $Y^{X, \ikrep}$ as the principal's `control' and turning their problem into a more standard stochastic control problem, albeit one that still must deal with the complexity of the mean field distribution which shows up in the state dynamics of $X$. As before, we emphasize a subtle difference between the nature of the mean fields that show up in the agents' and principal's problems.

The principal sets their choice of $\ovc$ (equivalently $\bm{Y_T^{\ikrep}}$, for all $k \in \mcK$, up to a constant%\dena{was this specific to linear utility?}\note{No, because the agents problem dictates that $Y_T$ is the derivative of the penalty fuction at time $T$}
). This results in the McKean-Vlasov FBSDE \eqref{eqn:rec_pa_theory_FBSDE-full} for the agents' problem, which we know has a fixed point across the space of mean field measure flows. This means that there is a bivariate distribution, which, when taken as an input to the problem, results in optimal controls and hence controlled states that have joint law exactly equal to the distribution that was taken as an input. This results in consistency in the mean field game, as desired.

From the principal's perspective, they can view their choice of $\ovc$ as inducing the best response of the agents, which is guaranteed to result in a consistent solution to the mean field game. As such, when the principal optimizes, it is more accurate to consider the distribution of the mean field matched to the joint law of the states at each point in time. More specifically, this leads to a mean-field type problem for the principal. This is in contrast to the more myopic agents, who do not know nor care about the mean field distribution. The agents seek only to minimize their costs given what they see in their filtrations, with the knowledge that they are individually insignificant and do not have an impact on the market. Hence, from their perspective, viewing the mean field distribution as fixed until after the optimization is completely makes more sense. 

Therefore, from the principal's perspective, we can re-state the dynamics of $X$ and $A$ to be:

\begin{align}
dX_t^\ikrep &= \left(h_t^k - \upsilon^k\,  Y_t^{X,\ikrep}
+ \tfrac{1}{\gamma^k \eta} \sum_{j\in\mcK} \eta^j \int Y_t^{X,(j)} d\PP^{\bm{Y_t^{(j)}}} + A_t^\ikrep \right) \,dt + \sigma_t^k dW_t^k, \label{eq:rec_pa_theory_principal_RECs} \\
dA_t^\ikrep &= \frac{- Y_t^{X, \ikrep}}{\beta^k} dt \label{eq:rec_pa_theory_principal_capacity}\end{align}

In order to turn this problem into a more standard stochastic control problem, we consider the dynamic version of an agent's value function, which is given by the definition below.
\begin{definition} \label{def:val_fn_sde}
For each $k \in \mcK$, we make the following definitions, given controls $g^k, \Gamma^k$, mean field distribution $\bm{\theta}$, and penalty function $C^k$:
\begin{align}
    U^\ikrep(t, X_t^\ikrep, A_t^\ikrep) &:=  \EE_t\left[\int_t^T \left(\tfrac{\zeta^k}{2}(g_u)^2 + \tfrac{\gamma^k}{2}(\Gamma_u)^2 + \tfrac{\beta^k}{2}(\alpha_u)^2 + S_u^{\bm{\theta}} \Gamma_u \right)du + C^k(X_T^k)\right], \\ U_t^\ikrep &= U^\ikrep(t, X_t^\ikrep, A_t^\ikrep)
\end{align}

We call $U_t^\ikrep$ the agent's continuation utility. The continuation utility induced by the optimal controls is called the value function, and is denoted by $V_t^\ikrep$.
\end{definition}

It is well-known that (given $\bm{\theta}$) $V_t^\ikrep$ satisfies the following BSDE (see Lemma 4.47 of \cite{carmona2018probabilistic}\note{This lemma uses a boundedness assumption, so we may have to change this citation. We might also be able to remove the citation entirely, because we essentially prove this later}): %\dena{common penalty in eq (3.14).} \note{corrected}

\begin{align}
    dV_t^\ikrep &= -f(t, X_t^\ikrep, A_t^\ikrep, \bm{\theta}, g^{\star}(\bm{Y_t^\ikrep}), \Gamma^{\star}(\bm{Y_t^\ikrep}), \alpha^{\star}(\bm{Y_t^\ikrep})dt + \sigma^k Y_t^{X, \ikrep} dW_t^k \label{eq:val_fn_bsde} \\ V_T^\ikrep &= C^k(X_T^\ikrep) \label{eq:val_fn_bsde_TC}
\end{align}
where 
\begin{align}
    f(t, X_t^\ikrep, A_t^\ikrep, \bm{\theta}, g, \Gamma, \alpha) &:= \tfrac{\zeta^k}{2}(g_t)^2 + \tfrac{\gamma^k}{2}(\Gamma_t)^2 + \tfrac{\beta^k}{2}(\alpha_t)^2 + S_t^{\bm{\theta}} \Gamma_t,
\end{align}
and $g^\star, \Gamma^\star, \alpha^\star$ are the controls that optimize the representative agent's Hamiltonian function. This assumes that the value function exists in a classical sense, and is once differentiable in time and twice differentiable in states. This is a standing assumption throughout our work.
%\note{check}

Note that we could combine the forward equations \eqref{eq:rec_pa_theory_HetFBSDE_fwdX}-\eqref{eq:rec_pa_theory_HetFBSDE_fwdA} with \eqref{eq:val_fn_bsde} and obtain another FBSDE whose solution would characterize the optimal behaviour of the agent, given a mean field distribution $\bm{\theta}$. Accordingly, the solution to the McKean-Vlasov version of this composite FBSDE would solve the agents' problem in full, as the solution to the McKean-Vlasov version of \eqref{eqn:rec_pa_theory_FBSDE-full} does. We note that the $\bm{Y_t^\ikrep}$ term above is the same as the $\bm{Y_t^\ikrep}$ terms in the statement of \Cref{thrm:rec_pa_theory_mfg_soln}.

Therefore, we can use our knowledge from \Cref{thrm:rec_pa_theory_mfg_soln}, which tells us that optimal controls $(g^\ikrep, \Gamma^\ikrep, \alpha^\ikrep)_{k \in \mcK}$ as well as a mean field distribution $\bm{\theta}$ exist which obtain a Nash equilibrium across agents. We can modify \eqref{eq:val_fn_bsde} by substituting the optimal controls and equilibrium REC price obtained through \Cref{thrm:rec_pa_theory_mfg_soln}:

\begin{align}
    dV_t^\ikrep &= \biggl(- \tfrac{1}{2}\upsilon^k (Y_t^{X,\ikrep})^2 - \tfrac{1}{2} (\beta^k)^{-1} (Y_t^{A, \ikrep})^2 +  \tfrac{1}{2\gamma^k} \left(\frac{\sum_{j \in \mcK} \eta^j \int - (Y_t^{X, (j)})  d\PP^{\bm{Y_t^{(j)}}}}{\eta}\right)^2\biggr)dt \nonumber \\ & \qquad+ \sigma^k Y_t^{X, \ikrep} dW_t^k \label{eq:val_fn_bsde_opt}, \\
    V_T^\ikrep &= C^k(X_T^\ikrep). \label{eq:val_fn_bsde_opt_tc}
    \end{align} 

We consider the dynamic version of the agent's value function as we wish to reframe the principal's problem from choosing over penalty functions described by $\ovc$ to choosing over stochastic processes $Y_t^{X,\ikrep}$. In doing so, we could replace the $C^k(X_T^\ikrep)$ term in \eqref{eq:principal_cost} with $V_T^\ikrep$. However, in order for this to be rigorously justified, we must argue that any valid penalty function $C$ can be represented in the form  $V_T^\ikrep = C^k(X_T^\ikrep)$ for all $k \in \mcK$.

Consider the set of McKean-Vlasov SDEs given by the following:

\begin{align}
    \hat{X}_t^\ikrep &= \hat{X}_0^\ikrep + \int_0^t \left(h_r^k - \upsilon^k\,  \hat{Y}_r^{X, \ikrep}
+ \tfrac{1}{\gamma^k \eta} \sum_{j\in\mcK} \eta^j \int \hat{Y}_r^{X, (j)} d\PP^{\bm{\hat{Y}_r^{(j)}}} + \hat{A}_r^\ikrep \right) \,dr + \int_0^t \sigma^k dW_r^k \label{eq:MV_SDE_state} \\
\hat{A}_t^\ikrep &= \hat{A}_0^\ikrep + \int_0^t -\frac{\hat{Y}_r^{X, \ikrep}}{\beta^k} dr \label{eq:MV_SDE_stateA}
\\
\hat{V}_t^{\ikrep, \hat{Y}^\ikrep} &= \hat{V}_0^\ikrep + \int_0^t \left(- \tfrac{1}{2}\upsilon^k (\hat{Y}_r^{X, \ikrep})^2 - \tfrac{1}{2}(\beta^k)^{-1}(\hat{Y}_t^{A, \ikrep})^2+ \tfrac{1}{2\gamma^k} \left(\frac{\sum_{j \in \mcK} \eta^j \int - (\hat{Y}_r^{(j)})  d\PP^{\bm{\hat{Y}_r^{(j)}}}}{\eta}\right)^2\right)dr +  \nonumber \\ &\qquad \qquad \int_0^t \sigma^k \hat{Y}_r^\ikrep dW_r^k. \label{eq:MV_SDE_val}
\end{align}

%\note{modify notation}
We define the set $\hat{\Xi}$ as follows:

%\note{This proof relies on being able to take a solution to an SDE and convert it into a solution to a BSDE - ensure this is valid}

\begin{align}
    \hat{\Xi} := \{ &(V_T^{\ikrep, \bm{Y^\ikrep}})_{k \in \mcK}: \forall k \in \mcK ,\; V_t^{\ikrep, \bm{Y_t^\ikrep}}, \bm{Y_t^\ikrep} \text { solves } \eqref{eq:MV_SDE_state}-\eqref{eq:MV_SDE_val}, \EE[V_0^\ikrep] \leq R_0,\, \\  &\EE[(V_T^{\ikrep, Y^\ikrep})^2] < \infty,\nonumber V_T^{\ikrep, \bm{Y^\ikrep}} = G(X_T^{(k)}) \text { s.t. $G$ convex, $G'$ exists,} \nonumber \\
    & G^\prime \text{ is continuous, Lipschitz and bounded, and $G''$ exists a.e.}\,\} \nonumber \label{eq:ncc_fun_set_alt_theory}
\end{align} 

We wish to show that this is in fact equivalent to the set $\Xi$.\dena{maybe $G^k$?} 

\begin{proposition} \label{prop:reduction}
For $\Xi$ as defined through \eqref{eq:ncc_fun_set_theory} and $\hat{\Xi}$ defined through \eqref{eq:ncc_fun_set_alt_theory}, $\Xi = \hat{\Xi}$.
\end{proposition} 
\begin{proof}
We will show this is true by showing that $\Xi \subset \hat{\Xi}$ and vice-versa.

$\bm{\Xi \subset \hat{\Xi}}$:

First, consider $\chi \in \Xi$. $\chi$ can be expressed as $\chi = (C^k(X_T^\ikrep))_{k \in \mcK}$, where each $C^k$ is convex, continuously differentiable with bounded and Lipschitz derivative, and with a second derivative which exists a.e. We also have that $\EE[C^k(X_T^\ikrep)^2] < \infty$. 

Consider the agents' MFG problem with $\ovc = (C^1(\cdot), \cdots, C^K(\cdot))$ defining the non-compliance functions they face. As covered in \cref{thrm:rec_pa_theory_mfg_soln}, we know this MFG has a solution, as the agents' cost functional is convex and G\^ateaux differentiable. Therefore, a control is optimal if and only if it makes the G\^ateaux derivative of the agents' cost functional vanish. This occurs when the controls are as in \eqref{eq:rec_pa_theory_optG}-\eqref{eq:rec_pa_theory_optAlpha}, with them being fully specified by the solution to the MV-FBSDE described by \eqref{eqn:rec_pa_theory_FBSDE-full}. Moreover, we know a unique solution to this exists, again by \Cref{thrm:rec_pa_theory_mfg_soln}. 

Consider the optimally controlled value function of a representative agent from sub-population $k$ from time $0$, denoted by $V^{a, \ikrep}(C^k)$. We can also consider this indexed through time, as $V_t^\ikrep$, or in it's Markov representation $V^\ikrep(t, x, a)$, representing the cost to the agent if they behave optimally starting at time $t$ with REC inventory $x$ and generation capacity $a$. Naturally, $V_T^\ikrep = C^k(X_T^\ikrep)$. By It\^o's Lemma, the dynamics of the value function process are:

\begin{align}
    dV^\ikrep_t &= (\partial_t V^\ikrep + (h_t^k + g^{\ikrep, \star}_t + \Gamma^{\ikrep, \star}_t + A_t^\ikrep) \partial_x V^\ikrep + \alpha_t^{\ikrep, \star} \partial_A V^\ikrep +  \tfrac{1}{2}\sigma^k \partial_{x x} V^\ikrep) dt \\ + &\sigma \partial_x V dW_t^k, \qquad \qquad {V}_T^\ikrep = C^k({X}_T^\ikrep). \label{eq:vf_dynamics}
\end{align}

Here, the starred controls represent the optimals. From \Cref{thrm:rec_pa_theory_mfg_soln}, we know the form of these optimals in terms of the stochastic process $\bm{Y_t^\ikrep}$ and mean field distribution $\bm{\theta}$. 

\textcolor{black}{Under sufficient regularity of $V^\ikrep$, we obtain that $\partial_x V^\ikrep(t, \tilde{X}_t^\ikrep, \tilde{A}_t^\ikrep) = \tilde{Y}_t^{\tilde{X}, \ikrep}$ and $\partial_A V^\ikrep(t, \tilde{X}_t^\ikrep, \tilde{A}_t^\ikrep) = \tilde{Y}_t^{\tilde{A}, \ikrep}$, where $\tilde{X}, \tilde{A}$ represent the controlled state processes under the optimal controls (see for example, Theorem 6.4.7 in \cite{pham2009continuous}\note{This result isn't exactly what we need... ideally, we want to show $D V^\ikrep = Y$ in our situation. I think there's a result in XYZ (theorem 4.1) that may work, but it seems too restrictive}).}

We know such a $\bm{\tilde{Y_t}^\ikrep}$, along with an associated mean field distribution $\bm{\tilde{\theta}}$ exists by \Cref{thrm:rec_pa_theory_mfg_soln}, and can use the DPE along with the optimal controls given by \eqref{eq:rec_pa_theory_optG} - \eqref{eq:rec_pa_theory_optAlpha} in \eqref{eq:vf_dynamics} to obtain:

\begin{align}
    dV_t^\ikrep = &\biggl[ -\tfrac{1}{2}\upsilon^k(\tilde{Y}_t^{X, \ikrep})^2 - \tfrac{1}{2}(\beta^k)^{-1}(\tilde{Y}_t^{A, \ikrep})^2 + \tfrac{1}{2}(\gamma^k)^{-1}(S_t^{\bm{\tilde{\theta}}})^2\biggr] \nonumber \\ &+ \sigma^k \tilde{Y}_t^{X, \ikrep} dW_t^k.
\end{align}

Integrating, we see that

\begin{align}
    V_t^\ikrep = V_T^\ikrep + &\int_t^T\biggl( \tfrac{1}{2}\upsilon^k(\tilde{Y}_r^{X, \ikrep})^2  + \tfrac{1}{2}(\beta^k)^{-1}(\tilde{Y}_r^{A, \ikrep})^2 - \tfrac{1}{2}(\gamma^k)^{-1}(S_r^{\bm{\tilde{\theta}}})^2\biggr)dr \\ &+ \int_t^T \sigma^k \tilde{Y}_r^{X, \ikrep}dWr^k. \label{eq:controlled_vf_dynamics} \nonumber
\end{align}

By substituting in $V_T^\ikrep = C^k(\tilde{X}_T^\ikrep)$ and taking time $t$ conditional expectation, we find

\begin{align}
    V_t^\ikrep = \EE_t\biggl[C^k(\tilde{X}_T^\ikrep) + &\int_t^T\biggl( \tfrac{1}{2}\upsilon^k(\tilde{Y}_r^{X, \ikrep})^2  + \tfrac{1}{2}(\beta^k)^{-1}(\tilde{Y}_r^{A, \ikrep})^2 - \tfrac{1}{2}(\gamma^k)^{-1}(S_r^{\bm{\tilde{\theta}}})^2\biggr)dr\biggr]. \label{eq:controlled_vf}
\end{align}

If we set $t = 0$, this becomes

\begin{align}
    V_0^\ikrep = \EE_0\biggl[C^k(\tilde{X}_T^\ikrep) + &\int_0^T\biggl( \tfrac{1}{2}\upsilon^k(\tilde{Y}_r^{X, \ikrep})^2  + \tfrac{1}{2}(\beta^k)^{-1}(\tilde{Y}_r^{A, \ikrep})^2 - \tfrac{1}{2}(\gamma^k)^{-1}(S_r^{\bm{\tilde{\theta}}})^2\biggr)dr\biggr].
\end{align}

But the right hand side is exactly $V^{a, \ikrep}(C^k)$. Finally, we observe that if we take expectations with respect to $\theta_0^\ikrep$, we get $\EE[V_0^\ikrep] = \int V^{a, \ikrep}(C^k) d\theta_0^\ikrep \leq R_0$. 

Consider the $V_T^\ikrep$ which emerges from \eqref{eq:controlled_vf_dynamics}. By construction, we see that it is the time $T$ value of the stochastic process $V_t^\ikrep$ which satisfies the appropriate McKean-Vlasov SDEs \eqref{eq:MV_SDE_state}-\eqref{eq:MV_SDE_val} (along with the processes $\tilde{X},\tilde{A}, \tilde{Y},  \tilde{Z}$, and subbing in the equilibrium REC price in McKean-Vlasov form).\note{Is there an issue with how the SDE is stated with respect to the REC price... sometimes we write it as dependent on $\bm{\mu}$ and sometimes we write it in the McKean-Vlasov form where it is integrated with respect to $\bm{Y}$.} 

This holds for all $k \in \mcK$. We also can see that $V_T^\ikrep = C^k(X_T^\ikrep)$ with $C^k$ inheriting convexity. It also inherits that its first derivative is continuous, bounded, and Lipschitz. Finally, it inherits existence of second derivative a.e., and integrability. Finally, as $\EE[V_0^\ikrep] = \int V^{a, \ikrep}(C^k) d\theta_0^\ikrep \leq R_0$, we have that $V_T^\ikrep \in \hat{\Xi}$, as required. 
\note{modified the above paragraph for the corrected version of $\hat{\Xi}$ and $\Xi$, which reflects that Delarue 2002 applies to $C^\prime$, not $C$.} 

$\bm{\hat{\Xi} \subset \Xi:}$

Consider $(V_T^{\ikrep, \bm{Y_t^\ikrep}})_{k \in \mcK} \in \hat{\Xi}$. By definition, $\forall k \in$ we have that $V_T^{\ikrep, \bm{Y_t^\ikrep}} = G^k(X_T^\ikrep)$, where $G^k$ is convex, $G^{k \prime}$ is continuous, bounded, and Lipschitz, and $G^{k \prime \prime}$ exists a.e. \note{modified above paragraph to reflect proper use of Delarue 2002}

Moreover, we have that $V_t^{\ikrep, \bm{Y_t^\ikrep}}$ is a component of a solution to the MKV SDEs described in \eqref{eq:MV_SDE_state}-\eqref{eq:MV_SDE_val}, which is induced by the choice of $\bm{Y_t^\ikrep}$ . Therefore, we can choose $\bm{Y_t^\ikrep}$ such that they are part of the solution to the MV-FBSDE described by \eqref{eqn:rec_pa_theory_FBSDE-full}, with $G^{k \prime}$ \dena{$G^\prime$?}\note{changed this; it previously said G, incorrectly} as the terminal condition for $Y_t^{X, \ikrep}$. Let us refer to this as $\bm{\hat{Y}_t^\ikrep}$. This choice implies a set of optimal controls, optimally controlled states, and mean field distribution $\bm{\theta}$ whose law coincides with the law of the controlled states at all times. These implied processes become part of the solution to \eqref{eq:MV_SDE_state}-\eqref{eq:MV_SDE_val}, along with $V_t^{\ikrep, \bm{Y_t^\ikrep}}$. By the definition of $\hat{\Xi}$, we also know that $\EE[G^k(\hat{X}_t^\ikrep)^2] < \infty$.

$\bm{\hat{Y}_t^\ikrep}$ corresponds to the solution to the agents' MFG with $\overrightarrow{G}:=(G^1(\cdot), \cdots, G^K(\cdot))$ as the non-compliance penalty function. Crucially, we know that $\bm{\hat{Y}_t^\ikrep}$ exists by the assumed properties of $G^k$ (convexity, bounded continuous differentiability,  and second derivative existing everywhere in particular), which make it such that a solution to the agents' MFG is guaranteed. Therefore, consider $(V_T^{\ikrep, \bm{\hat{Y}_t^\ikrep}})_{k \in \mcK} \in \hat{\Xi}$.

Consequently, \eqref{eq:MV_SDE_val} represents the dynamic value process of a representative agent in the MFG with penalty function $G^k$\note{does this need more explanation}. Substituting appropriately, \eqref{eq:MV_SDE_val} becomes

\begin{align}
    V_0^{\ikrep, \bm{\hat{Y}_0^\ikrep}} &= G^k(\hat{X}_T^\ikrep) + \int_0^T \biggl( \tfrac{1}{2}\upsilon^k (\hat{Y}_t^{X, \ikrep})^2 + \tfrac{1}{2}(\beta^k)^{-1}(\hat{Y}_t^{A, \ikrep})^2 + \tfrac{1}{2}(\gamma^k)^{-1}(S_t^{\PP^{\bm{\hat{Y}_t}}})^2 \biggr)dt  \nonumber \\
    &\qquad\qquad\qquad  + \sigma^k \int_0^T Y_t^{X, \ikrep} dW_t^k. \\
    \implies \EE[V_0^{\ikrep, \bm{\hat{Y}^\ikrep}}] &= \int V^{a, \ikrep}(G^k)d\theta_0^{\ikrep} \leq R_0.
\end{align}

This holds for all $k \in \mcK$, with the inequality holding due to the assumption that $\EE[V_0^{\ikrep, \bm{\hat{Y}^\ikrep}}] \leq R_0$ for all $k \in \mcK$ due to the definition of $\hat{\Xi}$. 

Therefore, we have that $(V_T^{\ikrep, \bm{\hat{Y}^\ikrep}})_{k \in \mcK} \in \Xi$, as required. 
\end{proof}

Note that $V_T^{\ikrep, \bm{Y^\ikrep}}$ is induced by $\bm{Y_t^\ikrep}$ and  $V_0^{\ikrep, \bm{Y_0^\ikrep}}$; these will ultimately be the control processes of the principal. To ease the notational burden, we suppress the $\bm{Y_t^\ikrep} $ superscript going forward. This proposition allows us to express the principal's problem as the following: 
\note{need to mention $U_P$ is non-decreasing and convex; this gives us convexity w/r/t X, V}
\begin{equation}
    \inf_{V_0^\ikrep \leq R_0} \inf_{(\bm{Y_t^\ikrep})_{t \in \mfT,k\in\mcK} \in L^2} \EE\left[U_P\left(\sum_{k \in \mcK} \pi_k \left( - V_T^\ikrep  - \lambda^k X_T^\ikrep  \right) \right) \right], \label{eq:principal_problem_final}
\end{equation}

subject to \dena{for the reservation utility shouldn't we use integral?}

\begin{align}
    dX_t^\ikrep &= \left(h_t^k - \upsilon^k\,  Y_t^{X, \ikrep}
    + \tfrac{1}{\gamma^k \eta} \sum_{j\in\mcK} \eta^j \int Y_t^{X, (j)} d\PP^{\bm{Y_t^{\ikrep}}} + A_t^\ikrep \right) \,dt + \sigma_t^k dW_t^k. \label{eq:agent_recs_state} \\
    dA_t^\ikrep &= \tfrac{-Y_t^{X, \ikrep}}{\beta^k} dt \label{eq:agent_capacity_state} \\
    dV_t^\ikrep &= \left(- \tfrac{1}{2}\upsilon^k (Y_t^{X, \ikrep})^2 - \tfrac{1}{2}(\beta^k)^{-1}(Y_t^{A, \ikrep})^2 + \tfrac{1}{2\gamma^k} \left(\frac{\sum_{j \in \mcK} \eta^j \int - (Y_t^{X, (j)}) d\PP^{\bm{Y_t^{\ikrep}}}}{\eta}\right)^2\right)dt + \nonumber\\ & \qquad \sigma^k_t Y_t^\ikrep dW_t^k. \label{eq:agent_val_fn_state}
\end{align}

Note that we have the above state equations for each $k \in \mcK$, so the overall state space dimension is $3K$.
% \note{can ignore the $V_0$ aspect for now - can change $V_0$ as needed at the end by adding constant term to it}

We can separate \eqref{eq:principal_problem_final} into two separate choices: the choice over $V_0^\ikrep$, and the choice over $\bm{Y_t^\ikrep}$. We currently consider only the latter\note{unclear how $V_0$ actually gets determined to me}. 

When optimizing over the choice of $\bm{Y_t^\ikrep}$, we have a more standard (though still quite non-standard) optimal control problem. Before proceeding to solve it using extended McKean-Vlasov techniques discussed in \cite{acciaio2019extended}, we recap the intuition behind our general approach to this principal agent problem, and its recasting to a decision over a stochastic process instead of a penalty function. The principal's decision is as follows. They optimize over their choices of penalty functions, which is equivalent to optimizing over choices of $\bm{Y_t^\ikrep}$ for all $k \in \mcK$. This induces a best response for the agents that is a function of $\bm{Y_t^\ikrep}$. In our case, this is governed by \eqref{eq:rec_pa_theory_optG} and \eqref{eq:rec_pa_theory_optAlpha}. Through $\bm{Y_t^\ikrep}$, these controls feature in \eqref{eq:agent_recs_state}, \eqref{eq:agent_capacity_state}, and \eqref{eq:agent_val_fn_state}, which impact the performance criterion of the principal. $\bm{Y_t^\ikrep}$ is part of the solution to \eqref{eq:val_fn_bsde_opt}, which implicitly incorporates the optimal controls of the agent, and thus ties the choice of $\bm{Y_t^\ikrep}$ to an optimally controlled dynamic value process, and through the terminal condition of said process, back to the non-compliance penalty function $C^k$.

The principal's problem described by \eqref{eq:principal_problem_final}, \eqref{eq:agent_recs_state}, \eqref{eq:agent_capacity_state}, and \eqref{eq:agent_val_fn_state} problem falls into a class of optimal controlled problems of extended McKean-Vlasov type, due to the dependence of the state equations on the law of the controlled process (in this case, $Y_t^{X, \ikrep}$).

Broadly, our methodology is as follows. We express the Hamiltonian of the principal's problem, and aim to optimize it over the choice of $Y_t^\ikrep$, for all $k \in \mcK$ simultaneously. Loosely, the method by which we do so is by taking the Hamiltonian, differentiating it with respect to the control, and finding the roots (assuming the Hamiltonian is sufficiently convex). However, due to the dependence of the principal's state equations on the distribution of the control ($Y_t^{X, \ikrep}$ in particular), the derivative we take includes a derivative with respect to the relevant measure. More specifically, we seek to use the sufficient condition for optimality described by Theorem 3.5 in \cite{acciaio2019extended}. 

Using this, we obtain first order conditions for the principal's controls in terms of some adjoint stochastic processes. We can then use the stochastic maximum principle developed in \cite{acciaio2019extended} to express an FBSDE for these adjoints, the solution to which fully characterizes optimal control $\bm{Y_t^\ikrep}$. This optimal control corresponds to a penalty function $C^k$.

\subsection{The Principal's Hamiltonian}

As discussed, we plan on using \cite{acciaio2019extended} to find the optimal controls for the principal's problem. To do so, we must ensure the principal's problem satisfies Assumptions (I)-(II) in \cite{acciaio2019extended}. Assumption (I) is clearly satisfied, while Assumption (II) requires that we impose a boundedness condition on $\bm{Y_t} = (\bm{Y_t^\ikrep})_{k \in \mcK}$, and that we choose $U_P$ such that its derivative have at most linear growth. Utility functions of the form $u(x) = x^a$ where $a \in [1,2)$, and \textcolor{black}{insert other utility function here} are examples of valid utility functions for this. \note{please check this paragraph, and lets discuss whether we need to add more here}

We also require the Hamiltonian be convex. We make that more explicit in the following assumption.

\begin{assumption}[Convexity of the Hamiltonian] \label{assume:hamiltonian_convexity}
We make the assumption that the Hamiltonian  $H(t, x, y, \theta, k, l)$ is convex in $\theta$, where the notion of convexity is defined in (3.10) of \cite{acciaio2019extended}. \dena{I think it would be better if we write the condition for our case and say that we assume this holds according to [1].}
\end{assumption}

We begin by writing out the Hamiltonian, as in (3.1) of \cite{acciaio2019extended}.

\begin{equation}
    H(t, x, y, \theta, k, l) = b(t, x,\theta, y) \cdot k + \sigma(t, x, \theta, y) \cdot l,
\end{equation}
where $b(t, x, \theta, y)$ represents the vector of drifts and $\sigma(t, x, \theta, y)$ represents the matrix of volatilities described in \eqref{eq:agent_recs_state},\eqref{eq:agent_capacity_state} and \eqref{eq:agent_val_fn_state}, for all $k \in \mcK$. This means the state space has dimension $3K$. Our goal is to find the controls ($\bm{Y_t}$ such that the derivative of the Hamiltionian is 0. This control will be specified in terms of the controlled state vector, as well as adjoint stochastic processes which solve a particular BSDE. Assuming sufficient convexity of the Hamiltonian with respect to state, control, and measures, this gives us the optimal control to the principal's problem (see Theorem 3.5 in \cite{acciaio2019extended}).

Denote $\bm{\overrightarrow{X}_t}$ as the vector $((X_t^\ikrep)_{k\in\mcK}, (A_t^\ikrep)_{k\in\mcK}, (V_t^\ikrep)_{k\in\mcK})^T$. The principal's Hamiltonian evaluated at the relevant processes is:%\dena{Shouldn't $S^{p}$ appear with a positive sign in $H$?}\note{I don't think so. This is the dynamics of $X$ and the term with $S$ in it arises from the optimal trading rate, which should be inversely correlated with price, as it is}\dena{You're right! Had forgotten $S$ itself has a negative sign.}

\begin{align}
    H(t, &\bm{\overrightarrow{X}_t}, \bm{Y_t}, \PP^{\bm{Y_t}}, K_t, L_t) = \nonumber \\
    \sum_{k \in \mcK}\biggl[&\bigl( h_t^k - \frac{Y_t^{X, \ikrep}}{\zeta^k} + \frac{-Y_t^{X, \ikrep} - S_t^{{\PP^{\bm{Y_t}}}}}{\gamma^k} + A_t^\ikrep\bigr) K_t^{X, \ikrep} - \frac{Y_t^{A, \ikrep}}{\beta^k} K_t^{A, \ikrep} + \nonumber\\
    & \bigl(-\tfrac{1}{2}\nu^k(Y_t^{X, \ikrep})^2 - \tfrac{1}{2\beta^k}(Y_t^{A, \ikrep})^2 + \tfrac{1}{2\gamma^k}(S_t^{\PP^{\bm{Y_t}}})^2\bigr) K_t^{V, \ikrep} + (L_t^{V, \ikrep})_k \sigma^k Y_t^{X, \ikrep}\biggr],
    % &K_t^X \cdot \overline{h_t} + \begin{bmatrix} K_t^{X,(1)}\nu^1 & \cdots & K_t^{X, (K)} \nu^{K} \end{bmatrix} \bm{Y_t} + \\
    % &\begin{bmatrix} \tfrac{K_t^{X,1}}{\gamma^1} & \cdots & \tfrac{K_t^{X, K}}{\gamma^{K}} \end{bmatrix} \biggl(\tfrac{1}{\eta}\int\bm{Y_t}^T\diag(\eta^1, \cdots , \eta^{K)}) \begin{bmatrix} d\PP^{X_t^1} \\ \vdots \\ d\PP^{X_t^{K}} \end{bmatrix}\biggr) \mathds{1}_k - \\
    % &\tfrac{1}{2} \bm{Y_t}^T \diag( K_t^{V,1}\nu^1, \cdots, K_t^{V, K} \nu^{K}) \bm{Y_t} + \\
    % &\tfrac{1}{2} \begin{bmatrix} \tfrac{K_t^{V,(1)}}{\gamma^1} & \cdots & \tfrac{K_t^{V, (K)}}{\gamma^{K}} \end{bmatrix} \biggl(\tfrac{1}{\eta}\int\bm{Y_t}^T\diag(\eta^1, \cdots , \eta^{K)}) \begin{bmatrix} d\PP^{X_t^1} \\ \vdots \\ d\PP^{X_t^{K}} \end{bmatrix}\biggr)^2 \mathds{1}_k + \\
    % & \begin{bmatrix} L_t^{V,1}\sigma^1 & \cdots & L_t^{V, K}\sigma^{K} \end{bmatrix} \bm{Y_t} + \text{Extr.}
\end{align}
where $K_t = \begin{pmatrix} K_t^{X, (1)} \\ \vdots \\ K_t^{X, (K)} \\ K_t^{A, (1)} \\ \vdots \\ K_t^{A, (K)} \\ K_t^{V, (1)} \\ \vdots \\ K_t^{A, (K)}
\end{pmatrix}$, $L_t = \begin{pmatrix} (L_t^{X, (1)})_1 & \cdots & (L_t^{X, (1)})_K \\ \vdots \\ (L_t^{X, (K)})_1 & \cdots & (L_t^{X, (K)})_K  \\ (L_t^{A, (1)})_1 & \cdots & (L_t^{A, (1)})_K \\ \vdots \\ (L_t^{A, (K)})_1 & \cdots & (L_t^{A, (K)})_K  \\ (L_t^{V, (1)})_1 & \cdots & (L_t^{V, (1)})_K \\ \vdots \\ (L_t^{V, (K)})_1 & \cdots & (L_t^{V, (1)})_K 
\end{pmatrix}$ as a $3K \times K$ matrix.

We want to find the first order conditions of the Hamiltonian. That is, the values of $\bm{Y}$ for which (3.5) from \cite{acciaio2019extended} is solved. Following the template laid out in \cite{acciaio2019extended}, we obtain the following first order conditions. 
% {\color{black}
% \begin{equation}
%   -\upsilon^1 K_t^{X, (1)} - \upsilon^1 Y_t^{X, (1)} K_t^{V, (1)} + \sigma^1(L_t^{V, (1)})_1 + \tilde{\EE}\left[\frac{\eta^1}{\eta}\sum_{j\in\mcK} \frac{1}{\gamma^j}\left(\tilde{K}_t^{X, (j)} + \tilde{K}_t^{V, (j)}\sum_{j \in \mcK} \frac{\eta^j}{\eta} \EE[Y_t^{X, (j)}]\right)\right] =0 \\
%     \vdots 
% \end{equation}
% }
\begin{subequations}
\begin{align}
    -\upsilon^1 K_t^{X, (1)} - \upsilon^1 Y_t^{X, (1)} K_t^{V, (1)} + \sigma^1(L_t^{V, (1)})_1 + \tilde{\EE}\left[\frac{\eta^1}{\eta}\sum_{j \in \mcK} \frac{1}{\gamma^j}\left(\tilde{K}_t^{X, (j)} +  \tilde{K}_t^{V, (j)} \sum_{j\in\mcK} \frac{\eta^j}{\eta} \EE[Y_t^{X, (j)}]\right)\right] &=0 \\
    \vdots \\
    -\upsilon^K K_t^{X, (K)} - \upsilon^K Y_t^{X, (K)} K_t^{V, (K)} + \sigma^1(L_t^{V, (K)})_1 + \tilde{\EE}\left[\frac{\eta^K}{\eta}\sum_{j \in \mcK} \frac{1}{\gamma^j}\left(\tilde{K}_t^{X, (j)} +  \tilde{K}_t^{V, (j)} \sum_{j\in\mcK} \frac{\eta^j}{\eta} \EE[Y_t^{X, (j)}]\right)\right] &=0 \\ 
    \vdots \\ 
    \frac{1}{\beta^1}\left(-K_t^{A, (1)} - K_t^{V, (1)}Y_t^{A, (1)}\right) &= 0\\
    \vdots \\ 
    \frac{1}{\beta^K}\left(-K_t^{A, (K)} - K_t^{V, (K)}Y_t^{A, (K)}\right) &= 0\\.
\end{align}
\label{eq:FOCs}
\end{subequations}

We must also express the BSDE the adjoint processes must solve, as in (3.2) of \cite{acciaio2019extended}. These can be described by

\begin{equation}
    dK_t = \begin{pmatrix} 0 \\
    \vdots \\
    0 \\
    -K_t^{X, (1)} \\
    \vdots \\
    -K_t^{X, (K)} \\
    0 \\
    \vdots\\
    0
    \end{pmatrix}dt + L_t dW_t, \qquad K_T = \begin{pmatrix} -U^\prime_P\left(-\sum_{k\in\mcK} \pi_k \left[\lambda^kX_T^\ikrep + V_T^\ikrep\right]\right)\pi_1 \lambda^1\\
    \vdots \\
    -U^\prime_P\left(-\sum_{k\in\mcK} \pi_k \left[\lambda^k X_T^\ikrep + V_T^\ikrep\right]\right)\pi_K \lambda^K \\
    0 \\
    \vdots \\
    0 \\
    -U^\prime_P\left(-\sum_{k\in\mcK} \pi_k \left[\lambda^kX_T^\ikrep + V_T^\ikrep\right]\right)\pi_1 \\
    \vdots \\
    -U^\prime_P\left(-\sum_{k\in\mcK} \pi_k \left[\lambda^k X_T^\ikrep V_T^\ikrep\right]\right)\pi_K
    \end{pmatrix}.
    \label{eq:adjoint_BSDE}
\end{equation}

We can trivially solve \eqref{eq:adjoint_BSDE} and obtain the following:

\begin{align}
    L_t &= 0 \; \forall t \\
    K_t^{X, (j)} &= - \pi_j M_t \lambda^j \\
    K_t^{V, (j)} &= - \pi_j  M_t \\
    K_t^{A, (j)} &= - \pi_j M_t \lambda^j (T-t) \\
    M_t &:= \EE_t\left[U^\prime_P\left(-\sum_{k \in \mcK} \pi_k \left[\lambda^k X_T^\ikrep + V_T^\ikrep\right]\right)\right].
\end{align}

\dena{add more details on derivation of $K_t^A$?}By substituting these into the first order conditions described by \eqref{eq:FOCs}, we obtain the following optimal controls: \dena{The sign of the second term is still negative?! I think it should be positive. Moreover there should be no $\lambda$ in the denominator.}\note{addressed now; I must have not made the change I thought I did previously}

\begin{align}
    Y_t^{A, \ikrep} &= -\lambda^k(T-t) \label{eq:optimal_yA} \\
    Y_t^{X, \ikrep} &= -\lambda^k + \frac{\tilde{\EE}[\frac{\eta^k}{\eta} \tilde{M}_t]\left(\sum_{j \in \mcK} \eta^j \lambda^j +  \sum_{j \in \mcK} \eta^j \EE[Y_t^{X, (j)}]\ \right)}{\upsilon^k M_t \pi_k}. \label{eq:optimal_yX}
\end{align}

The latter represents a fixed point of $Y_t^{X, \ikrep}$, but we can see that it is trivially solved by $Y_t^{X, \ikrep} = -\lambda^k$. This ultimately implies :
\begin{align}
    S_t^{\PP^{\bm{Y_t}}} &= \sum_{j \in \mcK}\frac{\eta^j}{\eta} \EE[-Y_t^{X, (j)}] \\ &=  \sum_{j \in \mcK}\frac{\eta^j \lambda^j}{\eta} %\frac{\frac{-1}{\lambda} - \EE[\frac{1}{M_t}\left(\sum_{j\in\mcK}\frac{(\eta^j)^2 \tilde{\EE}[\tilde{M}_t]}{\eta \upsilon^j, \pi_j \lambda}\right)}{1 - \EE[\frac{1}{M_t}]\left(\sum_{j\in\mcK}\frac{(\eta^j)^2 \tilde{\EE}[\tilde{M}_t]}{\eta \upsilon^j, \pi_j \lambda}\right)}
\end{align}

%\note{How best to explain this?}

Under \Cref{assume:hamiltonian_convexity}, Theorem 3.5 of \cite{acciaio2019extended} gives us that the optimizer of the Hamiltonian corresponds to the optimal control of the principal's problem, which we have now attained. This convexity also tells us that the optimizer of the Hamiltonian is unique, and occurs when the FOCs are $0$ (that is, at the values in \eqref{eq:optimal_yA}-\eqref{eq:optimal_yX}. 

Due to \eqref{eqn:rec_pa_theory_FBSDE-full}, we also know that $C^k(X_T^\ikrep) = -\lambda^k X_T^k + \text{ const.}$ and in fact, the principal will choose the same penalty function \dena{It's not the same penalty function} for all sub-populations. 

From these results, we can comment on the interesting economic implications for these market-based emissions regulation systems, under our assumptions. A linear non-compliance penalty function implies a constant marginal value of obtaining an additional REC of $\lambda^k$ for a firm in sub-population $k$. This implies that the REC price is a weighted average of these $\lambda^k$'s as stated above. Through \eqref{eq:rec_pa_theory_optG} and \eqref{eq:rec_pa_theory_optGamma}, we can see that this further implies a constant excess REC generation and REC trading rate. In particular, firms will buy or sell a constant, pre-visible amount based on whether the value to them of an additional REC is higher or lower than the average of $\lambda^k$'s, weighted by $\frac{\eta^j}{\eta}$.

We can make similar inferences about capacity expansion. Recall that $Y_t^{A, \ikrep}$ can be interpreted as the expected marginal gain the agent accrues for an incremental increase in generation capacity. This does not depend on the cost of expansion $\beta^k$ and is independent of the sub-population that the agent belongs to. Hence, the amount that an individual agent gains from the acquisition of another REC, or through additional REC capacity is always constant. Of course, the agents' optimal expansion (described by \eqref{eq:rec_pa_theory_optAlpha}) does incorporate the cost of expansion that the agent faces. However, it is completely deterministic, for the same reason as detailed in the previous paragraph.

In that sense, these results suggest that the optimal REC market, regardless of the utility function of the principal, isn't a market at all! Rather, it is more akin to a tax and rebate system, where the rebate would be determined through $V_0^\ikrep$, and chosen such that the tax does not put undue strain on the regulated agents. Ultimately, this implies that the most economically efficient and ideal way to institute a market-based system is to impose a tax instead, which is a novel and unexpected result. 

Naturally, this analysis is incomplete, and does not incorporate political realities, or the ease with which one can calibrate each of these systems. These remain important areas to explore from a policy perspective in future work. Nonetheless, the optimality of a tax-like system instituted within the confines of a REC market is incredibly interesting and a novel result that we think may have important implications for how emissions regulation systems should be instituted going forward.

% \subsubsection{Linear Utility}

% In the case of linear utility, we can follow the steps outlined above, with the utility function $U_p(x) = x$. In doing so, we obtain the following results for the principal's optimal controls:

% \begin{align}
%     Y_t^{A, \ikrep} &= - \frac{T - t}{\lambda} \label{eq:linear_utility_yA} \\
%     Y_t^{X, \ikrep} &= - \frac{1}{\lambda}%-\frac{\left( 1 + \lambda \sum_{j \in \mcK} \frac{\eta^j}{\eta} \EE[Y_t^{X, (j)}]\right)}{1 + \frac{\gamma^k}{\zeta^k}} \label{eq:linear_utility_yX} \\
%     \frac{\eta^j}{\eta} \EE[-Y_t^{X, (j)}] &= \frac{1}{\lambda} \frac{\frac{1}{\lambda} + \left(\sum_{j \in \mcK} \frac{\pi_j}{\eta (\gamma^j)^2 \lambda}\right)}{1 + \sum_{j \in \mcK} \frac{\pi_j}{\eta (\gamma^j)^2 \lambda}}. \label{eq:linear_utility_price}
% \end{align}

% There are numerous interesting details that arise out of this. Recall that $Y_t^{A, \ikrep}$ can be interpreted as the expected marginal gain the agent accrues for an incremental increase in generation capacity. Under linear utility (or indeed, any utility), this does not depend on the cost of expansion $\beta^k$ and is independent of the sub-population that the agent belongs to. 

\section{Conclusion} \label{sec:conclusion}
In this work, we build off \cite{shrivats2020mean}, flipping the point of view to now incorporate the potential goals of a regulator in a REC market. We  focus on the principal agent problem where the principal aims to introduce a REC market to induce investment into some renewable energy generation. Their goals are to maximize revenue and said investment into renewable energy. 

Meanwhile, the agents each aim to navigate the market at minimum cost by modulating their planned excess REC generation, trading rate, and capacity expansion rate. By taking the infinite-player limit of the model, we can apply results from MFG theory to solve the agents' cost minimization problem. Given a solution to the agents' cost minimization problem we then consider how the principal may optimize the design of the REC market they institute for their own goals. \textcolor{black}{We impose a constraint of the principal having to ensure the average agent achieves a cost no worse than some exogenous reservation cost, to ensure that the market they choose is `fair' in some sense to the agents.}\note{not sure sentence in red is needed, as we dont do a lot with the reservation cost} 

We find the optimal penalty function for the principal to impose under these conditions, utilizing tools from the nascent field of PA-MFGs. In particular, we restrict ourselves to considering penalty functions for which the MFG the agents experience have a solution. The methods used in this process marry the ideas of \cite{sannikov2008continuous}, \cite{cvitanic2018dynamic}, and \cite{elie2019tale}, with the extended McKean-Vlasov control techniques espoused by \cite{acciaio2019extended}. Ultimately, we are able to find an optimal penalty function, under the critical assumption detailed in \Cref{assume:hamiltonian_convexity} - that the principal's Hamiltonian is convex with respect to its arguments, as well as the law of the states and controls, with the notion of convexity for the latter as defined in \cite{acciaio2019extended}.

In particular, we find that the optimal penalty function is linear in the terminal RECs of the agents (plus some constant). This implies a constant marginal benefit to REC acquisition, and therefore fixes the REC price at a deterministic constant. This is an interesting and unintuitive result that would have fascinating implications for the design of REC markets generally, as it implies that the optimal REC market (under these conditions and assumptions) is more akin to a tax than a market. 

However, there are notable areas for improvement and further work. The most obvious would be to prove \Cref{assume:hamiltonian_convexity}. However, introducing a common noise for the regulated agents would also add to the realism of the model, as would an analysis from a policy perspective to better understand the tradeoffs in emissions markets or taxes which are not explicitly included in our model formulation. Perhaps the most important addition to make would be to extend this model to capture a multi-period REC market.

Other possible extensions include using deep learning methods to numerically solve for the PA game (as in \cite{campbell2021deep} or \cite{aurell2020optimal}). 

Nonetheless, in providing the mathematical framework contained in this paper, we have  produced a valid extension of a coherent structure under which REC markets can be studied holistically, including from the perspective of a regulator. In particular, this perspective while incorporating the mean field aspect of agent behaviour represented a gap that was not yet filled in this area. We feel this is potentially of great use to regulatory bodies in these systems, and can be a springboard to further studies and analysis in the future.

\bibliographystyle{siamplain}
\bibliography{references}
\end{document}